\documentclass[pra,preprint,amsfonts,amsmath,amssymb,superscriptaddress,showpacs,12pt]{revtex4-1}

\usepackage{graphicx}
\usepackage{bm}
\usepackage[bookmarks=true,colorlinks=true,linkcolor=black,citecolor=black,filecolor=black,menucolor=black,urlcolor=black]{hyperref}
\usepackage{amsthm}
\usepackage{mathrsfs}
\usepackage{comment}
\usepackage{color}
\usepackage{ulem}

\DeclareMathOperator{\tr}{tr}

\newcommand{\leftidx}[3]{%
  {\vphantom{#2}}#1#2#3%
}

\newcommand{\bra}[1]{\langle #1 |}
\newcommand{\ket}[1]{| #1 \rangle}
\newcommand{\lrbra}[1]{\left\langle #1 \right|}
\newcommand{\lrket}[1]{ \left| #1 \right\rangle}
\newcommand{\braket}[2]{\langle #1 | #2 \rangle}
\newcommand{\mean}[1]{\langle #1 \rangle}

\newcommand{\binomial}[2]{\begin{pmatrix} #1 \\ #2 \end{pmatrix}}

\newcommand{\hrho}{\hat{\rho}}
\newcommand{\hsigma}{\hat{\sigma}}
\newcommand{\hA}{\hat{A}}

\newcommand{\hE}{\hat{E}}
\newcommand{\hI}{\hat{I}}
\newcommand{\hM}{\hat{M}}

\newcommand{\hX}{\hat{X}}
\newcommand{\hO}{\hat{O}}

\newcommand{\hU}{\hat{U}}
\newcommand{\ha}{\hat{a}}

\newcommand{\hn}{\hat{n}}

\newcommand{\equref}[1]{Eq.~(\ref{#1})}

\theoremstyle{definition}
\newtheorem{theo}{Theorem}
\newtheorem{lemm}{Lemma}

\begin{document}



\title{
Classicality condition on a system's observable in a quantum measurement 
and relative-entropy conservation law
}

\author{Yui Kuramochi} \author{Masahito Ueda}
\affiliation{Department of Physics, University of Tokyo, 7-3-1 Hongo, Bunkyo-ku, Tokyo 113-0033, Japan}
\date{\today}
\begin{abstract}
We consider the information flow on a system's observable $X$
corresponding to a positive-operator valued measure
under a quantum measurement process $Y$
described by a completely positive instrument
from the viewpoint of the relative entropy.
We establish a sufficient condition for the relative-entropy conservation law
which states that the averaged decrease in the relative entropy of the system's observable $X$
equals the relative entropy of the measurement outcome of $Y$,
i.e. the information gain due to measurement.
This sufficient condition is interpreted as
an assumption of classicality 
in the sense that
there exists 
a sufficient statistic in a joint successive measurement of $Y$ followed by $X$
such that the probability distribution of the statistic coincides with that of
a single measurement of $X$ for the pre-measurement state.
We show that in the case when $X$ is 
a discrete projection-valued measure and $Y$ is discrete,
the classicality condition is equivalent to the relative-entropy conservation
for arbitrary states.
The general theory on the relative-entropy conservation is applied to typical quantum
measurement models,
namely
quantum non-demolition measurement, 
destructive sharp measurements on two-level systems,
a photon counting, 
a quantum counting,
homodyne and heterodyne measurements.
These examples except for the non-demolition and photon-counting measurements
do not satisfy
the known Shannon-entropy conservation law proposed by 
Ban~(M. Ban, J. Phys. A: Math. Gen. \textbf{32}, 1643 (1999)),
implying that our approach based on the relative entropy is applicable to a wider class of 
quantum measurements.
\end{abstract}

\maketitle

\section{Introduction}
In spite of the inevitable state change by a quantum measurement process,
some quantum measurement models are known to conserve the information
about a system's observable.
Examples of such measurements in optical systems
include the quantum 
non-demolition 
(QND) measurement~\cite{guerlin2007progressive}
and the destructive 
photon-counting
measurement~\cite{doi:10.1080/713820643,0954-8998-1-2-005,PhysRevA.41.4127}
on a single-mode photon number.
In the QND measurement, the number of photons is not destructed
and the classical Bayes rule holds for the photon-number distributions
of pre- and post-measurement states.
On the other hand,  the photon-counting measurement 
is a destructive measurement on the system's photon-number
but we can still construct the photon number distribution
of the pre-measurement state
from the number of counts and the photon number of the post-measurement state.

This kind of information-conserving quantum measurement was discussed by
Ban~\cite{Ban1997209,int.jour.theor.phys.37.2491,0305-4470-32-9-012,0305-4470-32-37-304}
quantitatively in terms of the mutual information 
$I_{\hrho}(X:Y)$ between a system's observable $X$
described by a positive-operator valued measure (POVM)
and the measurement outcome of a completely positive (CP) instrument 
$Y$~\cite{davieslewisBF01647093,Kraus1971311,kraus10.1007/3-540-12732-1,:/content/aip/journal/jmp/25/1/10.1063/1.526000}.
Ban established a condition for $X$ and $Y$ under which
the following Shannon 
entropy~\cite{shannon1948amathematical,*shannon1948bmathematical} 
conservation law holds:
\begin{equation}
	I_{\hrho}(X:Y) = H_{\hrho}(X)  - E_{\hrho}[ H_{\hrho_y}(X)  ],
	\label{ban_conservation}
\end{equation}
where 
$\hrho$ is the pre-measurement state,
$\hrho_y$ is the post-measurement state
conditioned on the measurement outcome $y$,
$E_{\hrho}[\cdot]$ denotes the ensemble average over the measurement outcome $y$
for given $\hrho$,
and
$H_{\hrho}(X)$ is the Shannon entropy computed from
the distribution of $X$ for state $\hrho$.
The left-hand side of \equref{ban_conservation}
is the information gain about the system's observable $X$
which is obtained from the measurement outcome $Y$,
while the right-hand side is a
decrease in the uncertainty about the distribution of $X$ 
due to the state change of the measurement.
The physical meaning of the condition for the Shannon entropy 
conservation~(\ref{ban_conservation})
due to Ban
is, however, not clear. 
There are also measurement models with continuous outcomes 
in which 
information about a system's observable is conserved 
but
the Shannon entropy conservation~(\ref{ban_conservation})
does not hold
due to a strong dependence of the continuous Shannon entropy, 
or differential entropy, on a reference measure of the probability measure.
In this sense, 
it is difficult to regard Eq.~(\ref{ban_conservation}) as the quantitative
expression of the information conservation about $X$.

In this paper, we investigate the information flows of the measured observable
based on the relative entropies~\cite{kullbackleibler1951}
of the measurement process $Y$ 
and the observable $X$.
Operationally the consideration of the relative entropies corresponds to the situation when
the pre-measurement state is assumed to be prepared 
in one of the two candidate states, $\hrho$ or $\hsigma$,
and the observer infers from the measurement outcome $Y$
which state is actually prepared.
This kind of information
is quantified as relative entropy of $Y$ between $\hrho$ and $\hsigma$.
The same consideration applies to $X$ and we can define the relative entropy
of $X$ for candidate states $\hrho$ and $\hsigma$ in a similar manner.
Thus we can compare these relative entropies
as Ban did to the Shannon entropy and mutual 
information~\cite{int.jour.theor.phys.37.2491,0305-4470-32-9-012}.

The primary finding of this paper is 
Theorem~\ref{th_rent_conservation1}
which states that
a kind of classicality condition for $X$ and $Y$
implies 
the relative-entropy conservation law
which states that 
the relative entropy of the measurement outcome $Y$
is equal to the ensemble-averaged decrease in
the relative entropy of the system with respect to
the POVM $X$.
The classicality condition for $X$ and $Y$
assumed in Theorem~\ref{th_rent_conservation1}
can be interpreted as the existence of a sufficient 
statistic~\cite{halmos1949application,kullbackleibler1951} 
in a joint successive measurement of $Y$ followed by $X$
such that the distribution of the statistic coincides with 
that of $X$ for the pre-measurement state.
This condition permits a classical interpretation of the measurement process $Y$
in the sense that there exists a classical model 
that simulates 
the conditional change
of the probability distribution of $X$ 
in the measurement process $Y$
computed from system's density operator.
It is also shown that
the conservation of the relative entropy~(\ref{rent_conservation_gen})
holds in a wider range of quantum measurements than 
the Shannon-entropy conservation law~(\ref{ban_conservation})
since the relative entropy is free from the dependence on the reference measure as
in the Shannon entropy.


This paper is organized as follows.
In Sec.~\ref{sec:gen},
we show
the relative-entropy conservation law as Theorem~\ref{th_rent_conservation1}
under a classicality condition for a system's POVM $X$ 
and a measurement process $Y$.
A special case in which $X$ is projection-valued is
formulated in Theorem~\ref{th_diag_conservation}.
By further assuming the discreteness of both the projection-valued measure $X$
and the measurement outcome of $Y$,
we establish in Theorem~\ref{th_diag_equi} the equivalence
between 
the relative-entropy conservation law for arbitrary candidate states
and the classicality condition assumed in Theorem~\ref{th_diag_conservation},
i.e. the classicality condition is a necessary and sufficient condition 
for the relative-entropy conservation law in this case.
In Sec.~\ref{sec:examples},
we show that typical quantum measurements
satisfy the classicality condition,
which are quantum non-demolition measurements,
destructive sharp measurements on two-level systems,
photon-counting measurement,
quantum-counter measurement,
balanced homodyne measurement,
and heterodyne measurement.
In these examples excepet for the quantum non-demolition
and photon-counting measurements,
we show that
the Shannon-entropy conservation law~(\ref{ban_conservation})
does not hold.
In Sec.~\ref{sec:conclusion}, 
we summarize the main results of this paper.

\section{Relative-entropy conservation law}
\label{sec:gen}
In this section we consider
a quantum system described by a Hilbert space $\mathcal{H}$,
a system's POVM $X$
and measurement process $Y$ described by a CP instrument.
Here we assume that
$X$ is described by a density 
$\{ \hE^X_x \}_{x \in \Omega_X}$
of POVM with respect to a reference measure $\nu_0 (dx)$
and that $Y$ is described by a density of CP instrument 
$\{ \mathcal{E}^Y_y \}_{y \in \Omega_Y}$ with respect to 
a reference measure $\mu_0 (dy)$.
The probability densities for the measurement outcomes for $X$ and $Y$
for a given density operator $\hrho$
are given by
\begin{gather}
	p^X_{\hrho} (x) 
	= 
	\tr[\hrho \hE_x^X]
	\notag 
	\intertext{and}
	p^Y_{\hrho} (y)
	= \tr [\mathcal{E}^Y_y (\hrho) ]
	= \tr[ \hrho \hE^Y_y ],
	\label{gen_eydef}
\end{gather}
respectively,
where $\hE^Y_y = {\mathcal{E}^Y_y}^\dagger (\hI)$ is the density 
of the POVM
for the measurement outcome $y$, $\hI$ is the identity operator, 
and the adjoint $\mathcal{E}^\dagger$ of a superoperator $\mathcal{E}$ is defined by
$ \tr [  \hrho \mathcal{E}^\dagger (\hA) ] := \tr [ \mathcal{E} (\hrho) \hA  ]   $ 
for arbitrary $\hrho$ and $\hA$.
The post-measurement state for a given measurement outcome $y$ of $Y$ is given by
\begin{align}
	\hrho_y 
	=
	\frac{ \mathcal{E}^Y_y (\hrho)   }{   P^Y_{\hrho}(y)   } .
	\label{ypost}
\end{align}

The densities of POVMs $\hE^X_x$ and $\hE^Y_y$ satisfy the following completeness conditions:
\begin{align}
	\int \mu_0 (dy) \hE^Y_y &= \hI,
	\label{y_completeness}
	\\
	\int \nu_0(dx) \hE_x^X &= \hI .
	\label{xcompleteness}
\end{align}

As the information content of the measurement outcome,
we consider the relative entropies of the measurement outcomes for
$X$ and $Y$ given by
\begin{align}
	D_X(\hrho||\hsigma) 
	&:=
	D(p^X_{\hrho} || p^X_{\hsigma})
	\notag \\
	&= \int \nu_0(dx) 
	p^X_{\hrho} (x)
	\ln \left(
	\frac{ p^X_{\hrho} (x) }{ p^X_{\hsigma} (x) }
	\right),
	\label{xrentdef}
\end{align}
and
\begin{align}
	D(p^Y_{\hrho}||p^Y_{\hsigma}) 
	= \int \mu_0 (dy) p^Y_{\hrho} (y)
	\ln \left(
	\frac{ p^Y_{\hrho} (y) }{ p^Y_{\hsigma} (y) }
	\right)
	\label{yrentdef}
\end{align}
respectively.
The relative entropies in Eqs.~(\ref{xrentdef}) and (\ref{yrentdef})
are information contents obtained from the measurement outcomes
as to which state $\hrho$ or $\hsigma$ is initially prepared.

The main goal of the present work is to establish a condition
for $X$ and $Y$ such that the relative-entropy conservation law
\begin{equation}
	D(p^Y_{\hrho}||p^Y_{\hsigma})
	=
	D (p^X_{\hrho} || p^X_{\hsigma}) 
	- E_{\hrho}[  D (p^X_{\hrho_y} || p^X_{\hsigma_y})  ],
	\label{rent_conservation_gen}
\end{equation}
holds.
Before discussing the condition for $X$ and $Y$ 
we rewrite \equref{rent_conservation_gen} in a more tractable form
as in the following lemma.
\begin{lemm}
\label{lem:eq}
Let $\{ \hE^X_x \}_{x \in \Omega_X}$
be a density of POVM with respect to a reference measure $\nu_0(dx)$
and let $\{ \mathcal{E}^Y_y \}_{y\in \Omega_Y}$ be a density of CP instrument
with respect to a reference measure $\mu_0 (dy)$.
Then the relative-entropy conservation law~(\ref{rent_conservation_gen})
is equivalent to
\begin{equation}
	D(\tilde{p}^{XY}_{\hrho} || \tilde{p}^{XY}_{\hsigma})
	=
	D(p^X_{\hrho}||p^X_{\hsigma}),
	\label{rent_conservation2}
\end{equation}
where $\tilde{p}^{XY} (x,y)$ is the probability distribution for
a successive joint measurement of $Y$ followed by $X$.
\end{lemm}
\begin{proof}
The joint distribution $\tilde{p}^{XY} (x,y)$ and the conditional probability distribution
$\tilde{p}^{X|Y}_{\hrho}(x|y)$ of $X$ under given measurement outcome $y$ are given by
\begin{align}
	\tilde{p}^{XY}_{\hrho} (x,y)
	=
	\tr[ \mathcal{E}^Y_y(\hrho) \hE^X_x  ]
	=
	\tr [ \hrho  {\mathcal{E}^Y_y}^\dagger ( \hE^X_x )  ]
	\notag 
\end{align}
and
\begin{align}
	\tilde{p}^{X|Y}_{\hrho} (x|y)
	:=\frac{ \tilde{p}^{XY}_{\hrho} (x,y)  }{ p^{Y}_{\hrho} (y)  }
	=p^X_{\hrho_y}(x) ,
	\label{tP_cond}
\end{align}
respectively.
In deriving \equref{tP_cond}, we used the fact that
the marginal distribution of $Y$ is given by \equref{gen_eydef}
and the definition of the post-measurement state in \equref{ypost}.
From the chain rule for the classical relative entropy 
(e.g. Chap.~2 of Ref.~\cite{cover2012elements}),
we have
\begin{align}
	D(\tilde{p}^{XY}_{\hrho} || \tilde{p}^{XY}_{\hsigma})
	&= 
	D( p^{Y}_{\hrho} || p^{Y}_{\hsigma})
	+
	E_{\hrho}[ D(\tilde{p}^{X|Y}_{\hrho} (\cdot | y )  || \tilde{p}^{X|Y}_{\hsigma}  (\cdot|y) ) ]
	\notag \\
	&=
	D(p^{Y}_{\hrho} || p^{Y}_{\hsigma})
	+
	E_{\hrho}[ D( p^X_{\hrho_y}  ||  p^X_{\hsigma_y}   ) ],
	\label{chainrule}
\end{align}
where we used \equref{tP_cond} in deriving the second equality.
The equivalence between Eqs.~(\ref{rent_conservation_gen}) and (\ref{rent_conservation2})
is now evident from \equref{chainrule}.
\end{proof}
Equation~(\ref{rent_conservation2}) indicates that
the information about $X$ contained in the original states $\hrho$ and $\hsigma$
is equal to the information obtained from the joint successive measurement of
$Y$ followed by $X$.

Now our first main result is the following theorem
on the relative-entropy conservation law:
\begin{theo}
\label{th_rent_conservation1}
Let $X$ be a density of POVM $\{ \hE^X_x  \}_{x \in \Omega_X}$ 
with respect to a reference measure $\nu_0(dx)$
and let
$Y$ be a density of an instrument
$\{ \mathcal{E}^Y_y \}_{y\in \Omega_Y} $
with respect to a reference measure $\mu_0 (dy)$.
Suppose that
$X$ and $Y$ satisfy the following conditions.
\begin{enumerate}
\item
POVM of $Y$ is the coarse-graining of $X$, 
i.e. there exists a conditional probability $p(y|x) \geq 0$ such that
\begin{equation}
	\hE_y^Y
	=
	\int \nu_0(dx)
	p(y|x) \hE_x^X
	\label{ass_pyx}
\end{equation}
with the normalization condition
\begin{equation}
	\int \mu_0 (dy) p(y|x) = 1.
	\label{norm_pyx}
\end{equation}
\item
There exist functions $\tilde{x} (x;y)$ and $q(x;y) \geq 0$
such that
\begin{align}
	{\mathcal{E}^Y_y}^\dagger (\hE^X_x)
	= q(x;y) \hE_{\tilde{x}(x;y)}^X
	\label{gen_cond}
\end{align}
for any $x$ and $y$.
\item
For any $y$ and any smooth function $F(x)$,
\begin{align}
	\int
	\nu_0 (dx) 
	q(x;y)
	F(\tilde{x} (x;y)) 
	=
	\int
	\nu_0 (dx) 
	p(y| x )
	F(x) .
	\label{gen_cond2}
\end{align}
\end{enumerate}
Then the relative-entropy conservation law~(\ref{rent_conservation_gen}) 
or (\ref{rent_conservation2}) holds.
\end{theo}
\begin{proof}
We prove \equref{rent_conservation2}.
By taking a quantum expectation of \equref{gen_cond}
with respect to $\hrho$,
we obtain
\begin{align}
	\tilde{p}^{XY}_{\hrho} (x,y)
	=
	q(x;y)
	p^X_{\hrho} 
	(
	\tilde{x} (x;y)
	),
	\label{ch5cond1p}
\end{align}
Equation~(\ref{ch5cond1p}) implies that,
from the factorization theorem for the sufficient statistic~\cite{halmos1949application},
the stochastic variable $\tilde{x} (x;y)$
is a sufficient statistic
of the joint successive measurement of $Y$ followed by $X$.
Let us denote the probability distribution function of $\tilde{x}(x;y)$
with respect to the reference measure $\nu_0$
as $p^{\tilde{X}}_{\hrho} (x)$.
From the definition of 
$p^{\tilde{X}}_{\hrho} (x)$
and the condition~(\ref{gen_cond2}),
for any function $F(x)$
we have
\begin{align}
	\int
	\nu_0(dx)
	p^{\tilde{X}}_{\hrho} (x)
	F(x)
	&=
	\int
	\nu_0(dx)
	\int
	\mu_0(dy)
	\tilde{p}^{XY}_{\hrho} (x,y)
	F(\tilde{x}(x;y))
	\notag
	\\
	&=
	\int
	\mu_0(dy)
	\int
	\nu_0(dx)
	p(y|x)
	p^X_{\hrho} (x)
	F(x)
	\notag 
	\\
	&=
	\int
	\nu_0(dx)
	p^X_{\hrho} (x)
	F(x),
	\notag
\end{align}
which implies that the probability distribution of $\tilde{x} (x;y)$
coincides with that of the single measurement of $X$.
Thus the condition~(\ref{gen_cond2}) ensures
\begin{align}
	p^{\tilde{X}}_{\hrho} (x)
	=
	p^X_{\hrho} (x).
	\label{ch5temp1}
\end{align}
From Eqs.~(\ref{ch5cond1p}) and (\ref{ch5temp1}),
we have
\begin{align}
	D(\tilde{p}^{XY}_{\hrho}  ||  \tilde{p}^{XY}_{\hsigma} )
	=
	D (p^{\tilde{X}}_{\hrho} || p^{\tilde{X}}_{\hsigma} )
	=
	D (p^{X}_{\hrho} || p^{X}_{\hsigma} ),
	\notag
\end{align}
where in deriving the first equality, 
we used the relative entropy conservation for the sufficient statistic
due to Kullback and Leibler~\cite{kullbackleibler1951}.
\end{proof}
The physical meaning of the conditions~(\ref{gen_cond}) and (\ref{gen_cond2})
is clear from Eqs.~(\ref{ch5cond1p}) and (\ref{ch5temp1});
the condition~(\ref{gen_cond}) implies that $\tilde{x} (x;y)$ is a sufficient statistic
for the joint successive measurement of $Y$ followed by $X$
and the condition~(\ref{gen_cond2}) ensures that
the distribution of $\tilde{x} (x;y)$ is equivalent to that
of $X$ for the pre-measurement state.

The assumptions 1, 2, 3 in Theorem~\ref{th_rent_conservation1}
are interpreted as a kind of classicality condition
as the proof uses only the classical probabilities.
In fact, a statistical model 
\begin{equation*}
	\tilde{p}
	(x_{\mathrm{in}} , y , x_{\mathrm{out}})
	=
	\delta_{
	x_{\mathrm{in}} , 
	\tilde{x} 
	( x_\mathrm{out} ; y) 
	}
	q(x_{\mathrm{out}};y) 
	p_{\hrho}^X (x_\mathrm{in})
\end{equation*}
with its sample space
$\Omega_X \times \Omega_Y \times \Omega_X$
reproduces all the probabilities that appear in the proof,
where $x_\mathrm{in}$ and $x_{\mathrm{out}}$ are the system's
values of $X$ before and after the measurement of $Y$, respectively,
and $y$ is the outcome of $Y$.
Here, we assumed the discreteness of $\Omega_X$ for simplicity,
but the same construction still applies to the continuous case.

In Ref.~\cite{0305-4470-32-9-012},
Ban proves the conservation for the Shannon entropy~(\ref{ban_conservation})
by assuming Eqs.~(\ref{ass_pyx}), (\ref{norm_pyx}), (\ref{gen_cond2}) and
\begin{equation}
	{\mathcal{E}^Y_y}^\dagger (\hE^X_x)
	= p(x|\tilde{x} (x;y)) \hE_{\tilde{x}(x;y)}^X
	\label{ban_gen_cond}
\end{equation}
for all $x$ and $y$.
The condition~(\ref{ban_gen_cond}) is stronger than our condition~(\ref{gen_cond})
since $q(x;y)$ is, in general, different from $p(x|\tilde{x} (x;y))$.
In some examples discussed in the next section,
we will show that
condition~(\ref{ban_gen_cond}) together with the Shannon entropy-conservation 
law~(\ref{ban_conservation}) does not hold,
whereas our condition for the relative-entropy conservation law~(\ref{rent_conservation_gen}) does.
This implies that our condition can be applicable to a wider range of quantum measurements.
Furthermore, for the case in which
$X$ is a projection-valued measure
and
labels $x$ and $y$ are both discrete,
we can show that condition~(\ref{ban_gen_cond}) is equivalent to
the condition that the post-measurement state is one of eigenstates of $X$
if the pre-measurement state is also one of them.
(See appendix~\ref{sec:app_ban} for detail).

Now we consider the case in which 
the reference POVM is a projection-valued measure (PVM)
$\hE^X_x$
which satisfies the following orthonormal completeness condition:
\begin{gather}
	\hE^X_x \hE^X_{x^\prime} = \delta_{x,x^\prime} \hE^X_x,
	\quad
	\sum_{x \in \Omega_X} \hE^X_x =\hI 
	\quad \mathrm{for \, discrete \, }x ;
	\label{cons1}
	\\
	\hE^X_x \hE^X_{x^\prime} = \delta(x-x^\prime) \hE^X_x,
	\quad
	\int_{\mathbb{R}}dx \hE^X_x =\hI 
	\quad \mathrm{for \, continuous \, }x,
	\label{cons2}
\end{gather}
where $\delta_{x,x^\prime}$ is the Kronecker delta
and $\delta (x-x^\prime)$ is the Dirac delta function.
If $\hE^X_x$ is written as $\ket{x} \bra{x}$,
the $X$-relative entropy
\begin{equation}
	D_{\mathrm{diag}}(\hrho||\hsigma)
	:=
	\begin{cases}
	\displaystyle
	\sum_{x \in \Omega_X} \bra{x} \hrho \ket{x} 
	\ln \left(
	\frac{ \bra{x} \hrho \ket{x}  }{ \bra{x} \hsigma \ket{x}  }
	\right),
	\\
	\displaystyle
	\int dx \bra{x} \hrho \ket{x} 
	\ln \left(
	\frac{ \bra{x} \hrho \ket{x}  }{ \bra{x} \hsigma \ket{x}  }
	\right),
	\end{cases}
	\notag 
\end{equation}
is called 
the diagonal-relative entropy.
For this reference PVM,
the condition for the relative-entropy conservation law
is relaxed as shown in the following theorem.

\begin{theo} \label{th_diag_conservation}
Let $\{  \mathcal{E}^Y_y \}_{y \in \Omega_Y}$ be a density of an instrument
with respect to a reference measure $\mu_0 (dy)$
and $\hE^X_x $ be a PVM with the 
completeness condition~(\ref{cons1}) or~(\ref{cons2}).
Suppose that
$X$ and $Y$ satisfy the condition~(\ref{gen_cond}) in Theorem~\ref{th_rent_conservation1}.
Then
there exists a unique positive function $p(y|x) $
satisfying Eqs.~(\ref{ass_pyx}) and (\ref{norm_pyx}).
Furthermore the
relative-entropy conservation law in~\equref{rent_conservation_gen} holds.
\end{theo}
\begin{proof}
For simplicity, we only consider the case in which
the label $x$ for the PVM is discrete.
The following proof can easily be generalized to continuous $X$ 
by replacing the sum $\sum_x \cdots$ with the integral $\int dx \cdots$
and the Kronecker delta $\delta_{x,x^\prime}$ with the Dirac delta function
$\delta (x-x^\prime)$.

The summation of \equref{gen_cond}
with respect to $x$ gives
\begin{align}
	\hE^Y_y
	&=
	\sum_{x \in \Omega_X}
	q(x;y) 
	\hE^X_{\tilde{x}(x;y)}
	\notag \\
	&=\sum_{x^\prime \in \Omega_X}
	\left(
	\sum_{x\in \Omega_X} \delta_{x^\prime , \tilde{x}(x;y)} q(x;y)
	\right)
	\hE^X_{x^\prime}.
	\label{pqdelta}
\end{align}
Therefore
\begin{equation}
	p(y|x)
	=
	\sum_{x^\prime  \in \Omega_X} \delta_{x , \tilde{x}(x^\prime;y)} q(x^\prime;y)
	\label{diag_pyx}
\end{equation}
satisfies Eq.~(\ref{ass_pyx}).
The uniqueness and the normalization condition~(\ref{norm_pyx}) for $p(y|x)$ 
follow from
Eq.~(\ref{pqdelta}) and 
the completeness condition~(\ref{y_completeness}) for $\hE^Y_y$
noting that $\{ \hE^X_x \}_{x\in \Omega_X}$ is linearly independent.

Next, we show the relative-entropy conservation law~(\ref{rent_conservation_gen}).
From Theorem~\ref{th_rent_conservation1},
it is sufficient to show 
the condition~(\ref{gen_cond2}).
For an arbitrary function $F(x)$
we have
\begin{align}
	\sum_{x \in \Omega_X} q(x;y) F(\tilde{x}(x;y))
	&=
	\sum_{x^\prime \in \Omega_X} 
	\left(
	\sum_{x\in \Omega_X} 
	\delta_{x^\prime , \tilde{x}(x;y)} q(x;y)
	\right)
	F(x^\prime)
	\notag \\
	&=
	\sum_{x \in \Omega_X} 
	p(y|x) F(x),
	\notag
\end{align}
where we used Eq.~(\ref{diag_pyx}) in the second equality.
Then the condition~(\ref{gen_cond2}) holds.
\end{proof}

Next,
we consider the case in which
$X$ is a discrete PVM $\{  \hE^X_x \}_{x \in \Omega_X}$ 
with the discrete complete orthonormal condition~(\ref{cons1})
and $Y$ is a discrete measurement 
on a sample space $\Omega_Y $
described by a set of CP maps $\{  \mathcal{E}^X_y \}_{y \in \Omega_Y}$
with the completeness condition
\begin{align}
	\sum_{y \in \Omega_Y}
	{ \mathcal{E}^Y_y }^\dagger (\hI)
	=
	\hI .
	\label{ch5completeness}
\end{align}
In this case,
we can show the equivalence between the established condition~(\ref{gen_cond})
in Theorem~\ref{th_diag_conservation} 
and the relative-entropy conservation law~(\ref{rent_conservation_gen}).

\begin{theo}
\label{th_diag_equi}
Let $X$ be a discrete PVM $\{  \hE^X_x  \}_{x \in \Omega_X}$
with a discrete complete orthonormal condition~(\ref{cons1})
and let $Y$ be a quantum measurement 
corresponding to a CP instrument on a discrete sample space 
$\Omega_Y$
described by a set of CP maps
$\{  \mathcal{E}^X_y \}_{y \in \Omega_Y}$
with the completeness condition~(\ref{ch5completeness}).
Then the following two conditions are equivalent:
\begin{enumerate}
\item[(i)]
The condition~(\ref{gen_cond}) holds
for all $x$ and $y$.
\item[(ii)]
The relative-entropy conservation 
law~(\ref{rent_conservation_gen}) 
or (\ref{rent_conservation2}) holds
for arbitrary states $\hrho$ and $\hsigma$.
\end{enumerate}
\end{theo}
To show the theorem,
we need the following lemma.
\begin{lemm}
\label{lemm:equi}
Let $\{ \hE^X \}_{x \in \Omega_X}$
be a PVM with a discrete complete orthonormal condition~(\ref{cons1})
and let $\{ \hE^Z_z  \}_{z \in \Omega_Z} $ be a discrete POVM.
Suppose that
\begin{align}
	D(p^X_{\hrho}||p^X_{\hsigma})
	=
	D(p^Z_{\hrho}||p^Z_{\hsigma})
	\label{ch5lemcond}
\end{align}
holds for any states $\hrho$ and $\hsigma$,
where
$p^X_{\hrho} (x) = \tr [\hrho \hE^X_x ]$ and
$p^Z_{\hrho} (z) = \tr [\hrho \hE^Z_z ]$.
Then for each $z \in \Omega_Z$ there exist a scalar $q(z) \geq 0$
and $\tilde{x} (z) \in \Omega_X$
such that
\begin{align}
	\hE^Z_z
	=
	q(z)
	\hE^X_{\tilde{x} (z)}.
	\label{ch5lem}
\end{align}
\end{lemm}

\begin{proof}[Proof of Lemma~\ref{lemm:equi}] 
Let $\hU_x$ be an arbitrary operator such that
$\hU_x^\dagger \hU_x = \hU_x \hU_x^\dagger = \hE^X_x$,
i.e. $\hU_x$ is an arbitrary unitary operator on a closed subspace 
$\hE^X_x \mathcal{H}$,
where $\mathcal{H}$ is the system's Hilbert space.
Define a CP and trace-preserving map $\mathcal{F}$ by
\begin{align*}
	\mathcal{F} (\hrho)
	:=
	\sum_{x\in \Omega_X}
	\hU_x 
	\hrho
	\hU_x^\dagger .
\end{align*}
Since 
$ 
\hE_{x}
\hU_{x^\prime}  
= 
\hE_x  
\hU_{x^\prime}   
\hU_{x^\prime}^\dagger
\hU_{x^\prime}
=
\hE_x  
\hE_{x^\prime}
\hU_{x^\prime} 
=
\delta_{x,x^\prime} \hU_{x^\prime},
$
we have
$p^X_{\hrho} (x) = p^X_{\mathcal{F}(\hrho)} (x)$
for any state $\hrho.$
Therefore, from the assumption~(\ref{ch5lemcond}) we have
\begin{align}
	D(p^Z_{\hrho}||p^Z_{\mathcal{F}(\hrho)})
	=
	D(p^X_{\hrho}||p^X_{\mathcal{F}(\hrho)})
	=
	0,
	\notag
\end{align}
and hence we obtain
\begin{align*}
	p^Z_{\hrho} (z)
	=
	p^Z_{\mathcal{F}(\hrho)} (z)
\end{align*}
for any $\hrho$ and any $z \in \Omega_Z$,
which is in the Heisenberg picture represented as
\begin{align}
	\hE^Z_z
	=
	{\mathcal{F}}^\dagger
	(\hE^Z_z) 
	=
	\sum_{x \in \Omega_X}
	\hU_x^\dagger
	\hE^Z_z
	\hU_x
	.
	\label{ch5lemtc1}
\end{align}
By taking $\hU_x$ as $\hE^X_x$,
we have
\begin{align}
	\hE^Z_z
	=
	\sum_{x \in \Omega_X}
	\hE_x^X
	\hE^Z_z
	\hE_x^X .
	\label{ch5lemtc2}
\end{align}
From Eqs.~(\ref{ch5lemtc1}) and (\ref{ch5lemtc2}),
an operator
$
	\hE_x^X
	\hE^Z_z
	\hE_x^X 
$
on $\hE_x^X \mathcal{H}$
commutes with an arbitrary unitary
$\hU_x$
on
$\hE_x^X \mathcal{H}$,
and therefore
$
	\hE_x^X
	\hE^Z_z
	\hE_x^X 
$
is proportional to the projection $\hE^X_x.$
Thus we can rewrite \equref{ch5lemtc2} as
\begin{align}
	\hE^Z_z
	=
	\sum_{x \in \Omega_X}
	\kappa (z|x)
	\hE^X_x ,
	\notag
\end{align}
where $\kappa (z|x)$ is a nonnegative scalar
that satisfies the normalization condition 
$\sum_{z \in \Omega_Z} \kappa (z|x) =1.$
Let us define a POVM $\{ \hE^{XZ}_{xz}  \}_{(x,z)  \in \Omega_X \times \Omega_Z}$
by
\begin{align}
	\hE^{XZ}_{xz}
	:=
	\kappa(z|x)
	\hE^X_x,
	\notag
\end{align}
whose marginal POVMs are given by $\hE^X_x$ and $\hE^Z_z$,
respectively.
Since the probability distribution for $\hE^{XZ}_{xz}$ is given by
\begin{align}
	p^{XZ}_{\hrho} (x,z)
	:=
	\tr[\hrho\hE^{XZ}_{xz}]
	=
	\kappa(z|x)
	p^X_{\hrho} (x),
\end{align}
$X$ is a sufficient statistic for 
a statistical model
$\{ p^{XZ}_{\hrho} (x,z) \}_{\hrho \in \mathcal{S}(\mathcal{H})}$,
where $\mathcal{S} (\mathcal{H})$ is the set of all the density operators
on $\mathcal{H}$.
Thus, from the sufficiency of $X$ and the assumption~(\ref{ch5lemcond}), we have
\begin{align*}
	D(p^{XZ}_{\hrho}||p^{XZ}_{\hsigma})
	=
	D(p^X_{\hrho}||p^X_{\hsigma})
	=
	D(p^Z_{\hrho}||p^Z_{\hsigma}).
\end{align*}
Since a statistic that does not decrease the relative entropy is
a sufficient statistic~\cite{kullbackleibler1951},
$Z$ is a sufficient statistic for 
$\{ p^{XZ}_{\hrho} (x,z) \}_{\hrho \in \mathcal{S}(\mathcal{H})}$.
Therefore there is a nonnegative scalar $r(x|z)$ such that
\begin{align*}
	p^{XZ}_{\hrho} (x,z)
	=
	r(x|z)
	p^Z_{\hrho} (z),
\end{align*}
or equivalently in the Heisenberg picture
\begin{align}
	\kappa(z|x) 
	\hE^X_x
	=
	r(x|z)
	\hE^Z_z .
	\label{ch5lemtc3}
\end{align}
To prove~(\ref{ch5lem}),
we have only to consider the case of $\hE^Z_z \neq 0.$
For such $z \in \Omega_Z$, there exists $x \in \Omega_X$
such that $\kappa(z|x) \hE^X_x \neq 0.$
Thus, from \equref{ch5lemtc3}
we have
$\hE^Z_z = \frac{\kappa(z|x)}{ r(x|z) } \hE^X_x$
and the condition~(\ref{ch5lem}) holds.
\end{proof}

\begin{proof}[Proof of Theorem~\ref{th_diag_equi}]
(i) $\Rightarrow$ (ii) is evident from Theorem~\ref{th_diag_conservation}.
Conversely, (i) readily follows from (ii) and Lemma~\ref{lemm:equi}
by identifying $\hE^Z_z$ with $ { \mathcal{E}^Y_y }^\dagger (\hE^X_x)$.
\end{proof}

\section{Examples of relative-entropy conservation law}
\label{sec:examples}
In this section, we apply the general theorem obtained in the previous section
to some typical quantum measurements,
namely a quantum non-demolition measurement,
a measurement on two-level sytems,
a photon-counting measurement,
a quantum-counter model,
homodyne and heterodyne measurements.

\subsection{Qunatum non-demolition measurement}
We first consider a quantum non-demolition (QND) 
measurement~\cite{RevModPhys.52.341,braginsky1995quantum,RevModPhys.68.1}
of a system's PVM $\ket{x}\bra{x}$.
In the QND measurement,
the $X$-distribution of the system is not disturbed by the measurement back-action.
This condition is mathematically expressed as
\begin{equation}
	p^X_{\mathcal{E}^Y(\hrho)}(x)
	=
	p^X_{\hrho}(x)
	\label{qndcond1}
\end{equation}
for all $\hrho$,
where 
\[
	\mathcal{E}^Y = \int \mu_0(dy) \mathcal{E}^Y_y 
\]
is the completely positive (CP) and trace-preserving map which describes
the state change of the system in the measurement of $Y$
in which the measurement outcome is completely discarded.
The QND condition in \equref{qndcond1} is also expressed 
in the Heisenberg representation as
\begin{equation}
	{\mathcal{E}^Y}^\dagger (\ket{x} \bra{x})
	= \ket{x} \bra{x}.
	\label{qndcond2}
\end{equation}
Let $\hM_{yz}$ be the Kraus operator~\cite{Kraus1971311}
of the CP map $\mathcal{E}^Y_y$ such that
\begin{equation*}
	\mathcal{E}^Y_y (\hrho) 
	= \sum_z \hM_{yz} \hrho \hM^\dagger_{yz}.
\end{equation*}
Then \equref{qndcond2} becomes
\begin{equation}
	\int \mu_0(dy) \sum_z 
	\hM^\dagger_{yz} 
	\ket{x} \bra{x}
	\hM_{yz}
	= \ket{x} \bra{x}.
	\label{qndcond3}
\end{equation}
Taking the diagonal element of \equref{qndcond3}
over the state $\ket{x^\prime}$ with
$x \neq x^\prime$, we have
\begin{equation*}
	\int \mu_0(dy) \sum_z | \bra{x} \hM_{yz} \ket{x^\prime} |^2
	= 0.
\end{equation*}
Therefore the Kraus operator $\hM_{yz}$ is diagonal in the $x$-basis
and, from \equref{ass_pyx}, it can be written as
\begin{equation}
	\hM_{yz} = 
	\begin{cases}
	\displaystyle
	\sum_x e^{i\theta(x;y,z)} \sqrt{p(y,z|x)} \ket{x} \bra{x},
	\\
	\displaystyle
	\int dx e^{i\theta(x;y,z)} \sqrt{p(y,z|x)} \ket{x} \bra{x},
	\end{cases}
	\label{hmy_qnd}
\end{equation}
where $p(y,z|x)$ satisfies
\[
	p(y|x) = 
	\sum_z  p(y,z|x).
\]
We take the reference PVM 
$\ket{x} \bra{x}$,
and from \equref{hmy_qnd} we have
\begin{equation}
	{\mathcal{E}^Y_y}^\dagger(\ket{x} \bra{x})
	=
	\sum_z
	\hM_{yz}^\dagger \ket{x} \bra{x} \hM_{yz}
	=
	p(y|x) \ket{x} \bra{x},
	\label{mxxm_qnd}
\end{equation}
which ensures the condition~(\ref{gen_cond})
with
\begin{align*}
	\tilde{x}(x;y)&=x, \\
	q(x;y) &= p(x|y).
\end{align*}
Thus from Theorem~\ref{th_diag_conservation}
the relative-entropy conservation law~(\ref{rent_conservation_gen}) holds.
In this case Ban's condition~(\ref{ban_gen_cond}) and 
Shannon entropy-conservation law in \equref{ban_conservation}
also hold~\cite{0305-4470-32-9-012}.

The relative entropy conservation relation in \equref{rent_conservation_gen}
in the QND measurement can be understood in a classical manner
as follows.
Let us consider a change in the $x$-distribution function
from $p^X_{\hrho}(x)$ 
to $p^X_{\hrho_y} (x)$.
In the QND measurement, by using \equref{mxxm_qnd},
the distribution of $X$ for the conditional post-measurement state becomes
\begin{equation}
	 p^X_{\hrho_y} (x) = \frac{ p(y|x)  p^X_{\hrho} (x)  }{ p^Y_{\hrho} (y) }.
	 \label{clas_bayes}   
\end{equation}
Note that the commutativity of $\ket{x}\bra{x}$ and $\hM_{yz}$ is essential
in deriving \equref{clas_bayes}.
Then \equref{clas_bayes} can be interpreted as
Bayes' rule for the conditional probability of $X$ under measurement outcome of $Y$.
Since the QND measurement does not disturb 
the system's observable $X$,
the change in the $X$-distribution of the system
is only the modification of observer's knowledge 
so as to be consistent with 
the obtained measurement outcome of $Y$
based on Bayes' rule in \equref{clas_bayes}.
Bayes' rule is also valid in a classical setup
in which the information about the system $X$ is
conveyed from the classical measurement outcome $Y$
without disturbing $X$. 
Since we can derive the relative-entropy conservation law
in \equref{rent_conservation_gen} 
from Bayes' rule in \equref{clas_bayes},
we can conclude that
the relative-entropy conservation law in both classical and QND measurements
is derived from the same Bayes' rule,
or the modification of the observer's knowledge.

The rest of this section is devoted to
exmaples of demolition measurements
in which the reference POVM observable $X$ is disturbed
by the measurement back-action,
yet the relative-entropy conservation law still holds.

\subsection{Measurements on two-level systems}
We consider a two-level system corresponding to a two-dimensional Hilbert space
spanned by complete orthonormal kets $\ket{0}$ and $\ket{1}$.
As the reference PVM of the system, we take
\begin{equation}
	\hE^X_x = \ket{x} \bra{x} 
	\quad (x=0,1).
\end{equation}
We consider a measurement $Y$ described by the following instrument:
\begin{align}
	\mathcal{E}^Y_y (\hrho)
	= \hat{\phi}_y 
	\bra{y} \hrho \ket{y}
	\quad 
	(y=0,1),
	\label{two_inst}
\end{align}
where $\hat{\phi}_y$ is an arbitrary state.
From \equref{two_inst} we can show that
\begin{align}
	\hE^Y_y
	&=
	\ket{y} \bra{y},
	\notag
	\\
	{\mathcal{E}^Y_y}^\dagger
	(\ket{x} \bra{x})
	&=
	\bra{x} \hat{\phi}_y \ket{x}
	\ket{y} \bra{y},
	\notag
\end{align}
or
\begin{align}
	p(y|x)
	&=
	\delta_{x,y},
	\label{two_pyx}
	\\
	q(x;y) 
	&=
	\bra{x} \hat{\phi}_y \ket{x},
	\label{two_q}
	\\
	\tilde{x}(x;y)
	&=
	y.
	\label{two_tx}
\end{align}
Then the conditions for Theorem~\ref{th_diag_conservation} are satisfied
and the relative-entropy conservation law
\begin{align*}
	D(p^Y_{\hrho}|| p^Y_{\hsigma})
	&=
	D_X(\hrho|| \hsigma)
	-
	E_{\hrho}[   D_X(\hrho_y|| \hsigma_y)  ]
	\\
	&=
	D_X(\hrho|| \hsigma)
\end{align*}
holds.
The second equality follows from $\hrho_y = \hsigma_y$.
On the other hand,
from Eqs.~(\ref{two_pyx})-(\ref{two_tx}),
Ban's condition~(\ref{ban_gen_cond}) does not hold
if the post-measurement state $\hat{\phi}_y$ 
does not coincide with one of eigenstates $\ket{x} \bra{x}$.

Let us examine the Shannon-entropy conservation law~(\ref{ban_conservation}).
To make the discussion concrete,
we assume $\hat{\phi}_y = \hI/2$.
Then the Shannnon entropy of $X$ and the mutual information between $X$ and $Y$
are evaluated to be
\begin{gather*}
	I_{\hrho} (X:Y)
	=
	H_{\hrho} (X)
	=
	- \sum_{x = 0,1} \bra{x} \hrho \ket{x} \ln \bra{x} \hrho \ket{x} ,
	\\
	H_{\hat{\rho}_y} (X) 
	=
	H_{\hat{\phi}_y} (X) 
	= \ln 2.
\end{gather*}
Thus
\begin{align}
	H_{\hrho} (X)
 	- H_{\hat{\rho}_y} (X)
 	= I_{\hrho}(X:Y) - \ln 2
 	\neq
 	I_{\hrho}(X:Y).
 	\notag
\end{align}
Therefore, the Shannon-entropy conservation law~(\ref{ban_conservation}) does not hold.
In this measurement model,
the measured information of $Y$ is maximal and 
any information is not contained in the post-measurement state.
This fact is properly reflected in the fact $D_X(\hrho_y||\hsigma_y)=0$
if we consider the relative entropy,
while the Shannon entropy is non-zero if the post-measurement state is an eigenstate.
This is the reason why the Shannon-entropy conservation law~(\ref{ban_conservation}) does not hold.

\subsection{Photon-counting measurement}
The photon-counting measurement
described in Refs.~\cite{doi:10.1080/713820643,0954-8998-1-2-005,PhysRevA.41.4127}
measures the photon number in a closed cavity 
in a destructive manner
and continuously in time.
The measurement process
in an infinitesimal time interval $dt$ 
is described by the following measurement operators:
\begin{align}
	\hM_0(dt) &= \hI - \left(i\omega + \frac{\gamma}{2} \right) \hn dt \label{pdmo0}, 
	\\
	\hM_1(dt) &= \sqrt{\gamma dt } \ha \label{pdmo1},
\end{align}
where $\omega$ is the angular frequency of the observed cavity photon mode,
$\gamma >0$ is the coupling constant of the photon field with the detector,
$\ha$ is the annihilation operator of the photon field,
and $\hn := \ha^\dagger \ha$ is the photon-number operator.
The event corresponding to the measurement operator in \equref{pdmo0}
is called the no-count process in which
there is no photocount,
while the event corresponding to \equref{pdmo1} is called
the one-count process
in which a photocount is registered.
In the one-count process,
the post-measurement wave function is multiplied 
by the annihilation operator $\ha$
which decreases the number of photons in the cavity by one.
Thus, this measurement is not a QND measurement.

From the mesurement operators for an infinitesimal time interval in Eqs.~(\ref{pdmo0}) and (\ref{pdmo1}),
we can derive an effective measurement operator for a finite time interval $[0,t)$ as follows 
(cf. Eq.~(29) in Ref.~\cite{0954-8998-1-2-005}):
\begin{equation}
	\hM_m(t) = \sqrt{ \frac{ (1-e^{- \gamma t})^m }{m!}  } 
	e^{- \left(i\omega + \frac{\gamma}{2} \right) t \hn } \ha^m,
	\label{pdmo2}
\end{equation}
where $m$ is the number of photocounts in the time interval $[0,t)$,
which corresponds to the measurement outcome $y$ in Sec.~\ref{sec:gen}.
The POVM for the measurement operator in \equref{pdmo2} can be written as
\begin{equation}
	\hM_m^\dagger(t) \hM_m(t) 
	=
	p(m|\hn;t),
	\label{pdpovm}
\end{equation}
where
\begin{equation}
	p(m|n;t) 
	=
	\binomial{n}{m}
	(1-e^{-\gamma t} )^m e^{-\gamma t (n - m)} .
	\label{pmnt}
\end{equation}
Equation~(\ref{pdpovm}) shows that the measurement outcome $m$
conveys the information about the cavity photon number $\hn.$
Especially in the infinite-time limit $t\rightarrow \infty$,
the conditional probability in \equref{pmnt}
becomes $\delta_{m,n}$,
indicating that the number of counts $m$ 
conveys the complete information about the 
photon-number distribution of the system.
Then we take the reference PVM as 
the projection operator into the number state,
$\ket{n} \bra{n}$,
with $\hn \ket{n} = n\ket{n}$ and the orthonormal
condition $\braket{n}{n^\prime} = \delta_{n,n^\prime}$.
From the measurement operator in \equref{pdmo2},
we obtain
\begin{gather}
	\hM_m^\dagger (t) \ket{n} \bra{n} \hM_m(t) 
	=
	q(n;m;t)
	\ket{\tilde{n}(n;m)} \bra{\tilde{n}(n;m)},
	\label{mnnm_pd}
	\\
	\tilde{n}(n;m)
	=
	n+m, 
	\label{n+m}
	\\
	q(n;m;t)  
	= p(m|m+n;t).
	\label{qnmt_pd}
\end{gather}
Equation~(\ref{n+m}) can be interpreted as
the photon number of the pre-measurement state
when the number of photocounts is $m$
and the photon number remaining in the post-measurement state is $n$.
From Eqs.~(\ref{mnnm_pd})-(\ref{qnmt_pd}),
the condition~(\ref{gen_cond})
for Theorem \ref{th_diag_conservation},
together with Ban's condition~(\ref{ban_gen_cond}),
is satisfied and
we have the relative entropy conservation relation for the photon-counting measurement
as
\begin{align*}
	D(p_{\hrho}(\cdot ;t) || p_{\hsigma} (\cdot;t))
	=
	D_{\mathrm{diag}}(\hrho||\hsigma)
	-
	E[D_{\mathrm{diag}}(\hrho_m(t)||\hsigma_m(t)) ],
\end{align*}
where $p_{\hrho}(m;t) = \tr [ \hrho \hM_m^\dagger(t) \hM_m(t) ]$
is the probability distribution of the number of photocounts $m$.
We remark that the Shannon entropy-conservation law in \equref{ban_conservation}
also holds in this measurement~\cite{Ban1997209}.

\subsection{Quantum counter model}
A quantum counter
model~\cite{PhysRevLett.68.3424,PhysRevA.53.3808}
is a continuous in time measurement on a signle-mode photon field
in which no-count and one-count measurement operators for an infinitesimal time interval
$dt$
are given by
\begin{align}
	\hM_0 (dt)
	&=
	\hI 
	- \frac{\gamma }{2} \ha \ha^\dagger dt,
	\notag
	\\
	\hM_1 (dt)
	&=
	\sqrt{\gamma dt} \ha^\dagger,
	\notag
\end{align}
respectively.
The effective measurement operator for a finite time interval
$[0,t]$
is known to be dependent only on the total number $m$ of counting events
in the time interval and given by~\cite{PhysRevA.53.3808}
\begin{align}
	\hM^{\mathrm{qc}}_m (t)
	=
	\sqrt{
	\frac{(e^{\gamma t} -1 )^m }{m!}
	}
	e^{  - \gamma t \ha \ha^\dagger /2 }
	\left(
	\ha^\dagger
	\right)^m .
	\label{qc_mop}
\end{align}
The POVM for this measurement is then
\begin{align}
	\hE^{\mathrm{qc}}_m (t)
	&=
	\hM^{\mathrm{qc}}_m {}^\dagger (t)
	\hM^{\mathrm{qc} }_m (t)
	\notag
	\\
	&=
	\frac{(e^{\gamma t} -1 )^m }{m!}
	\ha^m
	e^{  - \gamma t \ha \ha^\dagger  }
	(\ha^\dagger)^m
	\notag
	\\
	&=
	p^{\mathrm{qc}} (m|\hn ; t),
	\notag
\end{align}
where
\begin{align}
	p^{\mathrm{qc}} (m|n;t)
	&=
	\binomial{n+m}{m}
	(e^{\gamma t} -1)^m
	e^{-\gamma t (n+m +1)}
	\notag
\end{align}

In this measurement model we can show two kinds of relative-entropy conservation laws
corresoponding to two different system's observables.
As the first observable,
we take the PVM
$\ket{n} \bra{n}$.
Then from \equref{qc_mop},
we have
\begin{gather}
	\hM^{\mathrm{qc}}_m {}^\dagger (t)
	\ket{n} \bra{n}
	\hM^{\mathrm{qc} }_m (t)
	=
	p^{\mathrm{qc}} (m|\tilde{n}(n;m);t)
	\ket{\tilde{n}(n;m)}
	\bra{\tilde{n}(n;m)},
	\label{qc_joken1}
	\\
	\tilde{n}(n;m)
	=
	n-m
\end{gather}
and the conditions for Theorem~\ref{th_diag_conservation},
together with Ban's condition~(\ref{ban_gen_cond}), hold.
Therefore the relative-entropy conservation law
\begin{align}
	D ( p^{\mathrm{qc}}_{\hrho} (\cdot ;t) || p^{\mathrm{qc}}_{\hsigma} (\cdot ;t) )
	=
	D(p^N_{\hrho} || p^N_{\hsigma})
	- 
	E_{\hrho} 
	[   
	D(p^N_{\hrho_m(t)} || p^N_{\hsigma_m(t)})
	]
	\label{qc_rcons1}
\end{align}
holds, where
\begin{gather*}
	p^{\mathrm{qc}}_{\hrho} (m;t)
	=
	\tr \left[
		\hrho \hE^{\mathrm{qc}}_m (t)
	\right]
	=
	\sum_{n=0}^\infty 
	p^{\mathrm{qc}} (m|n;t) \bra{n} \hrho \ket{n} ,
	\\
	p^N_{\hrho}(n)
	=
	\bra{n} \hrho \ket{n},
\end{gather*}
with
$\hrho_m (t)$
being the post-measurement state when 
the measurement outcome is $m$.

The second system's POVM is given by
\begin{gather}
	\hE^X_{x} dx
	=
	p^X (x|\hn) dx,
	\label{qc_exdef}
	\\
	p^X (x|n)
	=
	\frac{e^{-x} x^n }{ n !},
	\notag
\end{gather}
where $x$ is a real positive variable.
The probability distribution of $X$ 
\begin{equation}
	p^X_{\hrho} (x) dx
	=
	\tr \left[
		\hrho
		\hE^X_{x}
	\right]
	dx
	\notag
\end{equation}
is known to be the distribution of 
$\lim_{t\rightarrow \infty} m /e^{\gamma t}$,
corresponding to the total information obtained during
the infinite time interval~\cite{PhysRevA.53.3808}.
Equation~(\ref{qc_exdef})
implies that $X$ is obtained by coarse-graining $\hn$.
It can be shown~\cite{PhysRevA.53.3808}
that the distribution $p^X_{\hrho} (x)$ determines
the photon-number distribution by
\begin{align}
	\left.
	\bra{n} \hrho \ket{n}
	=
	\frac{d^n}{dx^n}
	(e^x p^X_{\hrho} (x))
	\right|_{
	x = 0
	}
	.
	\notag
\end{align}
However,
this just implies that the Markov mapping
\begin{equation}
	p^X_{\hrho} (x)
	= 
	\sum_{n = 0}^\infty
	p^X(x|n)
	p^N_{\hrho} (n)
	\notag
\end{equation}
is injective and we cannot conclude that
the information contained in $X$ and $\hn$ are the same
as the following discussion shows.

From Eqs.~(\ref{qc_mop}) and (\ref{qc_exdef}) we obtain
\begin{align}
	\hM^{\mathrm{qc}}_m {}^\dagger (t)
	p^X (x|\hn)
	\hM^{\mathrm{qc} }_m (t)
	&=
	q(x;m) 
	p^X (\tilde{x} (x;m)  | \hn),
	\label{qc_2jouken1}
	\\
	q(x;m)
	&=
	e^{-\gamma t}
	p^{\mathrm{qc}}(m|\tilde{x}(x;m)),
	\label{qc2q}
	\\
	p^{\mathrm{qc}}(m|x)
	&=
	\frac{
		\left[( e^{\gamma t} -1  ) x \right]^m
	}{
		m!
	}
	\exp \left[
		-(e^{\gamma t} -1)x
	\right],
	\\
	\tilde{x}(x;m)
	&=
	e^{-\gamma t} x .
	\notag
\end{align}
Here
$
p^{\mathrm{qc}}(m|x)
$
satisfies
$
	\sum_{m =0}^\infty  p^{\mathrm{qc}}(m|x) =1.
$
Furthermore, for an arbitrary function
$F(x)$,
\begin{align}
	\int_0^\infty
	dx
	q(x;m)
	F(\tilde{x}(x;m))
	&=
	\int_0^\infty
	d( e^{-\gamma t} x)
	p^{\mathrm{qc}}(m|e^{-\gamma t}x)
	F(e^{-\gamma t}x)
	\notag
	\\
	&=
	\int_0^\infty
	dx
	p^{\mathrm{qc}}(m|x)
	F(x).
	\label{qc_2jouken2}
\end{align}
The POVM for the measurement outcome $m$ can be written as
\begin{align}
	\hM^{\mathrm{qc}}_m {}^\dagger (t)
	\hM^{\mathrm{qc} }_m (t)
	&=
	\int_0^\infty dx
	\hM^{\mathrm{qc}}_m {}^\dagger (t)
	p^X (x|\hn)
	\hM^{\mathrm{qc} }_m (t)
	\notag 
	\\
	&=
	\int_0^\infty dx
	q(x;m)
	p^X (\tilde{x} (x;m)  | \hn)
	\notag 
	\\
	&=
	\int_0^\infty dx
	p^{\mathrm{qc}}(m|x)
	p^X (x  | \hn).
	\label{qc_2jouken3}
\end{align}
From Eqs.~(\ref{qc_2jouken1}), (\ref{qc_2jouken2}) and (\ref{qc_2jouken3}) and
Thereom~\ref{th_rent_conservation1},
the relative-entropy conservation law
\begin{align}
	D ( p^{\mathrm{qc}}_{\hrho} (\cdot ;t) || p^{\mathrm{qc}}_{\hsigma} (\cdot ;t) )
	=
	D(p^X_{\hrho} || p^X_{\hsigma})
	-
	E_{\hrho} \left[
		D(p^X_{\hrho_m (t)} || p^X_{\hsigma_m (t)})
	\right].
	\label{qc_rcons2}
\end{align}
holds.

Let us consider the asysmptotic behaviors of relative entropies
in the limit
$t \rightarrow \infty$.
Since
$
m/e^{\gamma t}
$
converges to $X$ in distribution, we have
\begin{align}
	D ( p^{\mathrm{qc}}_{\hrho} (\cdot ;t) || p^{\mathrm{qc}}_{\hsigma} (\cdot ;t) )
	\xrightarrow{t \rightarrow \infty} 
	D(p^X_{\hrho} || p^X_{\hsigma}).
	\label{qc_limpqc}
\end{align}
From Eqs.~(\ref{qc_rcons1}), (\ref{qc_rcons2}) and (\ref{qc_limpqc})
we obtain
\begin{align}
	E_{\hrho} 
	[   
	D(p^N_{\hrho_m(t)} || p^N_{\hsigma_m(t)})
	]
	& \xrightarrow{t \rightarrow \infty} 
	D(p^N_{\hrho} || p^N_{\hsigma})
	-
	D(p^X_{\hrho} || p^X_{\hsigma}),
	\label{qc_lim1}
	\\
	E_{\hrho} \left[
		D(p^X_{\hrho_m (t)} || p^X_{\hsigma_m (t)})
	\right]
	& \xrightarrow{t \rightarrow \infty} 
	0.
	\label{qc_lim2}
\end{align}
From the chain rule of relative entropy~\cite{cover2012elements},
the right-hand-side of Eq.~(\ref{qc_lim1})
is evaluated to be
\begin{gather}
	\int_0^\infty dx
	p^X_{\hrho} (x)
	D( p^{N}_{\hrho} (\cdot |x) ||  p^{N}_{\hsigma}  (\cdot |x))
	\geq 0,
	\label{qcsa1}
\end{gather}
where
\begin{gather}
	p^{N}_{\hrho} (n |x)
	=
	\frac{ p^X(x|n)  p^{N}_{\hrho} (n) }{ p^X_{\hrho} (x) }
	\label{qc_pnx}
\end{gather}
is the photon-number distribution conditioned by $X$.
The equality in~(\ref{qcsa1})
holds if and only if the
photon-nuber distributions of
$\hrho$
and
$\hsigma$
conincide.
This can be shown as follows.

If the equality in
Eq.~(\ref{qcsa1})
holds,
we have
$D( p^{N}_{\hrho} (\cdot |x) ||  p^{N}_{\hsigma}  (\cdot |x)) = 0$
for almost all $x \geq 0$.
Thus
\begin{equation}
	\forall n \geq 0 , \quad
	p^N_{\hrho}(n|x)
	=
	p^N_{\hsigma}(n|x)
	\label{qc2tochu}
\end{equation}
for almost all
$x > 0$,
and therefore we can take at least one $x >0$
satisfying Eq.~(\ref{qc2tochu}).
From Eqs.~(\ref{qc_pnx}) and (\ref{qc2tochu}),
we have
\begin{align}
	\forall n \geq 0 , \quad
	\frac{
	\bra{n} \hrho \ket{n} 
	}{
	p^X_{\hrho} (x)
	}
	=
	\frac{
	\bra{n} \hsigma \ket{n} 
	}{
	p^X_{\hsigma} (x)
	}.
	\label{qc2tochu2}
\end{align}
Taking the summation of \equref{qc2tochu2} over
$n$,
we have
\begin{equation}
p^X_{\hrho} (x) = p^X_{\hsigma} (x).
\label{qc2tochu3}
\end{equation}
From Eqs.~(\ref{qc2tochu2}) and (\ref{qc2tochu3}),
we finally obtain
$\bra{n} \hrho \ket{n} = \bra{n} \hsigma \ket{n} $
$(\forall n\geq 0 ).$

Since the right-hand-side of Eq.~(\ref{qc_lim1}) is the difference between the information contents
of $\hn$ and $X$,
the above discussion shows that the measurement outcome $m$ 
carries strictly smaller information
than that contained in the photon-number distribution.
Equation~(\ref{qc_lim1})
also shows that the difference of these information contents are obtained by
a projection measurement on the post-measurement state.

From Eq.~(\ref{qc2q}) Ban's condition~(\ref{ban_gen_cond}) does not hold
for $X$.
The difference between the Shannon entropies of pre- and post-measurement states
is given by  
\begin{align}
	&H_{\hrho} (X)
	- E_{\hrho} [  H_{\hrho_m(t)} (X) ]
	\notag \\
	&=
	H_{\hrho} (X)
	+
	\sum_{m=0}^\infty
	p^{\mathrm{qc}}_{\hrho} (m)
	\int_0^\infty dx
	p^X_{\hrho_m(t)} (x)
	\ln p^X_{\hrho_m(t)} (x)
	\notag \\
	&=
	H_{\hrho} (X)
	+\sum_{m=0}^\infty
	\int_0^\infty dx
	e^{-\gamma t}
	p^{\mathrm{qc}}(m|e^{-\gamma t} x )
	p^X_{\hrho}(e^{-\gamma t} x)
	\ln \left(
	\frac{
		e^{-\gamma t}
		p^{\mathrm{qc}}(m|e^{-\gamma t}x)
		p^X_{\hrho}(e^{-\gamma t} x)
	}{
		p^{\mathrm{qc}}_{\hrho} (m)
	}
	\right)
	\notag \\
	&=
	-\gamma t 
	+ I_{\hrho} (X : \mathrm{qc})
	\neq
	I_{\hrho} (X : \mathrm{qc}),
	\label{qc_hineq}
\end{align}
and the Shannon-entropy conservation law~(\ref{ban_conservation}) does not hold.
The term $-\gamma t$ in Eq.~(\ref{qc_hineq}) comes from the Jacobian of the variable transformation
$x \rightarrow \tilde{x} (x;y) = e^{-\gamma t} x$
and the strong dependence of the Shannon entropy for a continuous variable on the reference measure 
$dx$.
On the other hand, if we take the relative entropy,
such dependence on the reference measure is absent 
and we can analyze both of information conservations of $\hn$ and $X$
in a consistent manner.

\subsection{Balanced homodyne measurement}
The balanced homodyne measurement~\cite{PhysRevA.47.642,1355-5111-8-1-015,CBO9780511813948}
measures one of the quadrature amplitudes 
of a photon field $\ha$ in a destructive manner
such that the system's photon field relaxes into a vacuum state $\ket{0}$.
This measurement process is implemented by
mixing the signal photon field with
a classical local-oscillator field
into two output modes
via a $50\%$-$50\%$ beam splitter
and taking the difference of the
photocurrents of the two output signals.
For later convenience,
we define the following quadrature amplitude operators:
\[
\hX_1:= \frac{\ha + \ha^\dagger}{\sqrt{2}} ,
\quad
\hX_2:= \frac{\ha - \ha^\dagger}{\sqrt{2}i}.
\]

The measurement operator in the interaction picture
for an infinitesimal time interval $dt$
is given by 
\begin{equation}
	\hM(d\xi(t) ;dt ) = \hI - \frac{\gamma}{2} \hn dt 
	+ \sqrt{\gamma} \ha \,  d \xi(t),
	\label{mop_homo}
\end{equation}
where
$\gamma$ is the stregth of the coupling with the detector,
$d \xi (t)$ is a real stochastic variable corresponding to
the output homodyne current
which satisfies the It\^o rule
\begin{equation}
	\left( d\xi(t) \right)^2 =dt.
	\label{ito_homo}
\end{equation}
The reference measure $\mu_0 (\xi(\cdot))$ for the measurement outcome 
is the Wiener measure in which 
infinitesimal increments
$\{ d\xi (s) \}_{s\in [0,t)}$
are independent Gaussian stochastic variables
with mean $0$ and variance $dt$.
From the measurement operator in~\equref{mop_homo},
the ensemble average of the outcome $d\xi(t)$ 
for the system's state $\hrho(t)$ at time $t$ is given by
\begin{equation}
	E[d\xi(t)|\hrho(t)] = \sqrt{2\gamma} \mean{\hX_1}_{\hrho(t)},
	\label{expectation_homo}
\end{equation}
where $\mean{\hA}_{\hrho} := \tr [\hrho \hA]$.
Equation~(\ref{expectation_homo})
indicates that $d\xi (t)$ measures the quadrature amplitude of the system.
The general properties of the continuous quantum measuerment 
with such diffusive terms
are investigated in Refs.~\cite{Wiseman200191,barchielli2009quantum}.

The time evolution of the system prepared in a pure state $\ket{\psi_0}$
at $t=0$ is given by the following stochastic Schr\"{o}dinger equation
\[
	\ket{\psi(t+dt)} = \hM(d\xi(t) ;dt ) \ket{\psi(t)}.
\]
The solution is given by~\cite{1355-5111-8-1-015}
\begin{equation}
	\ket{\psi(t)} = \hM_{y(t)}(t) \ket{\psi_0},
	\label{swf_homo}
\end{equation}
where
\begin{equation}
	\hM_{y(t)} (t) =e^{-\frac{\gamma t}{2} \hn }
	\exp \left[ y(t) \ha - \frac{1}{2} (1-e^{-\gamma t}) \ha^2 \right],
	\label{homomop}
\end{equation}
\begin{equation}
	y(t) = \sqrt{\gamma} \int_0^t e^{- \frac{\gamma s}{2} } d\xi(s).
	\label{y_homo}
\end{equation}
Note that $\ha^2$ term should be included 
in the exponent 
on the right-hand side of \equref{homomop}
to be consistent with the It\^{o} rule given in \equref{ito_homo}.
We also mention that
the measurement operator in \equref{homomop}
does not commute with the quadrature amplitude
operator $\hX_1$ and therefore
this measurement disturbs $\hX_1$.
In the infinite-time limit $t\rightarrow \infty$
the stochastic wave function in \equref{swf_homo}
approaches the vacuum state $\ket{0}$ 
regardless of the initial state,
which also indicates the destructive nature of the measurement.

As the reference PVM,
we take the spectral measure $ \ket{x}_1 {}_1 \bra{x} $
of the quadrature amplitude operator
$\hX_1$,
where $\ket{x}_1$ satisfies
\begin{equation}
	\hX_1 \ket{x}_1 = x \ket{x}_1,
	\quad
	{}_1\braket{x}{x^\prime}_1
	= 
	\delta(x-x^\prime).
	\notag
\end{equation}
Then, the operator
$
	\hM_{y(t)}^\dagger (t) 
	\ket{x}_1
	{}_1
	\bra{x} \hM_{y(t)}(t)
$
and the POVM for the measurement outcome $y(t)$
are evaluated to be
(see Appendix~\ref{sec:app1}
for derivation)
\begin{gather}
	\hM_{y(t)}^\dagger (t) 
	\ket{x}_1
	{}_1
	\bra{x}
	\hM_{y(t)}(t)
	=
	q(x;y(t);t)
	\lrket{
	\tilde{x}(x;y(t);t)
	}_1
	\leftidx{_1}{
	\lrbra{
	\tilde{x}(x;y(t);t) 
	}
	}{},
	\label{mxxm}
	\\
	q(x;y(t);t)
	=
	e^{-\gamma t/2} 
	p(y|\tilde{x}(x;y(t))),
	\label{homo_q}
	\\
	p(y|x)
	=
	\frac{1}{  
	\sqrt{ 2 \pi } 
	e^{-\gamma t } (1 - e^{-\gamma t } )
	}
	\exp \left[
		- \frac{
			\left(
				y - \sqrt{2}(1-e^{-\gamma t}) x
			\right)^2
		}{
			2 e^{-\gamma t } (1 - e^{-\gamma t } )
		}
	\right],
	\\
	\tilde{x}(x;y(t);t)
	=
	e^{-\frac{\gamma t}{2} } x 
	+ \frac{y(t)}{\sqrt{2}} ,
	\label{hom_xtilde}
	\\
	\mu_0(dy)\hM_y^\dagger (t) \hM_y (t)
	=
	dy p(y|\hX_1)
	\label{povm_homo}
\end{gather}
where the arguments of $\tilde{x}(x;y)$ in \equref{hom_xtilde}
are the the measurement outcome ($y(t)/\sqrt{2}$ on the right-hand side)
and the remaining signal of the system ($e^{-\frac{\gamma t}{2}} x$ on the right-hand side),
in which the exponential decay factor describes
the system's relaxation to the vacuum state 
and the loss of the initial information contained in the system.
The POVM in \equref{povm_homo}
shows that the measurement outcome $y(t)$
contains unsharp information about the quadrature amplitude $\hX_1$
and that in the infinite-time limit $t\rightarrow \infty$
the measurement reduces to the sharp measurement of $\sqrt{2} \hX_1$.

Equation~(\ref{mxxm})
indicates that the condition~(\ref{gen_cond})
for Theorem~\ref{th_diag_conservation} is satisfied,
and we obtain 
the relative-entropy conservation law
\[
	D(p^Y_{\hrho} (\cdot;t) ||p^Y_{\hsigma} (\cdot;t) )
	=
	D_{X_1}(\hrho||\hsigma)
	-
	E_{\hrho}[ D_{X_1}(\hrho_{y(t)} (t)||\hsigma_{y(t)} (t)) ],
\]
where 
\begin{equation*}
	p^Y_{\hrho}(y;t) dy
	=
	\tr [\hrho \hM_y(t)^\dagger \hM_y(t)]\mu_0(dy)
\end{equation*}
is the probability distribution function
of the measurement outcome $y(t)$ 
which is computed from the POVM in \equref{povm_homo},
$\hrho_{y(t)}(t)$ and $\hsigma_{y(t)}(t)$
are the conditional density operators for given measurement outcome
$y(t)$,
and $D_{X_1}(\hrho||\hsigma)$
is the diagonal relative entropy of 
the quadrature amplitude operator
$\hX_1$.

On the other hand, from Eq.~(\ref{homo_q}) Ban's condition~(\ref{ban_gen_cond}) does not hold.
The difference between the Shannon entropies is evaluated to be
\begin{align}
	&H_{\hrho}(X)
	-
	E_{\hrho} [  H_{\hrho_y}(X) ]
	\notag 
	\\
	&=
	H_{\hrho}(X)
	+
	\int dx dy
	e^{-\gamma t /2}
	p(y|\tilde{x} (x;y)) p^X_{\hrho} (\tilde{x}(x;y))
	\ln \left(
		\frac{
			e^{-\gamma t /2}
			p(y|\tilde{x} (x;y)) 
			p^X_{\hrho} (\tilde{x}(x;y))
		}{
			p^Y_{\hrho} (y)
		}
	\right)
	\notag
	\\
	&=
	-\frac{ \gamma t }{2}
	+ I_{\hrho} (X:Y)
	\neq
	I_{\hrho} (X:Y),
	\label{homo_gt}
\end{align}
and Shannon-entropy conservation law does not hold.
The term $-\gamma t/2$ in \equref{homo_gt} again arises from the non-unit Jacobian of the transformation
$x \rightarrow \tilde{x}(x;y)$ as in Eq.~(\ref{qc_hineq}).

\subsection{Heterodyne measurement}
The heterodyne measurement simultaneously
measures the two non-commuting quadrature amplitudes
$\hX_1$ and $\hX_2$
in a destructive manner
as in the homodyne measurement.
One way of implementation
is to take a large detuning of the local oscillator
in the balanced homodyne setup.
Then the cosine and sine components of 
the
homodyne current give
the two quadrature amplitudes~\cite{CBO9780511813948}.

The measurement operator for the heterodyne
measurement in an infinitesimal time interval $dt$ is given by
\begin{equation}
	\hM(d\zeta(t) ;dt ) = \hI - \frac{\gamma}{2} \hn dt + \sqrt{\gamma} \ha d \zeta(t),
	\label{mop_het}
\end{equation}
where $d\zeta(t)$ is a complex variable
obeying the complex It\^{o} rules
\begin{equation}
	(d \zeta (t))^2 = (d \zeta^\ast (t))^2 = 0,
	\quad
	d \zeta (t) d \zeta^\ast (t) = dt.
	\label{ito_comp}
\end{equation}
As in the homodyne measurement,
we consider the time evolution in the interaction picture.
The reference measure $\mu_0$ for 
the measurement outcome $\zeta(\cdot)$ is 
the complex Wiener measure
in which real and imaginary parts of $ d\zeta(\cdot)$
are statistically independent Gaussian variables 
with zero mean and second order moments 
consistent with the complex It\^{o} rules in \equref{ito_comp}.

The stochastic evolution of the wave function is 
described by the following stochastic Schr\"{o}dinger equation
\begin{equation}
	\ket{\psi(t+dt)} = \hM(dt;d\zeta (t)) \ket{\psi(t)}.
	\label{sse_hete}
\end{equation}
The solution of \equref{sse_hete}
for the initial condition $\ket{\psi_0}$
at $t=0$
is given by~\cite{1355-5111-8-1-015}
\[
	\ket{\tilde{\psi}(t)} 
	=\hM_{y(t)}(t)
	\ket{\psi_0},
\]
where	
\begin{gather}
	\hM_{y(t)}(t)
	=
	e^{- \frac{\gamma t}{2} \hn} e^{y(t) \ha},
	\label{mop_sol_het}
	\\
	y(t) = \sqrt{\gamma} \int_0^t e^{- \frac{\gamma s}{2} } d\zeta(s).
	\label{y_hete}
\end{gather}
Here the measurement operator in \equref{mop_sol_het}
does not involve the $\ha^2$ term
unlike the case of the homodyne measurement
in \equref{homomop}
because $(d\zeta(t))^2$ vanishes in this case.

Let us evaluate the POVM for the measurement outcome
$y(t)$ in \equref{y_hete}.
From \equref{mop_sol_het},
we have
\begin{align}
	&\hM_{y(t)}^\dagger(t) \hM_{y(t)}(t)
	\notag 
	\\
	&=
	\mathscr{A} \left\{
		\exp \left[
		\gamma t - (e^{\gamma t} -1) \ha \ha^\dagger 
		+ e^{\gamma t} (y(t) \ha + y^\ast (t) \ha^\dagger) 
		- e^{\gamma t} |y(t)|^2 
		\right]
	   \right\} ,
	\label{povm1_het}
\end{align}
where $\mathscr{A}\{ f(\ha,\ha^\dagger) \}$
denotes the antinormal ordering
in which the annihilation operators are
placed  to the left of the creation operators.
To obtain the proper POVM for the measurement outcome $y(t)$,
we have to multiply the operator $\hM_{y(t)}^\dagger(t) \hM_{y(t)}(t)$
by the measure $\mu_0(dy(t))$ 
which is the measure for the reference complex Wiener measure.
In the complex Wiener measure,
the variable $y(t)$ in \equref{y_hete}
is a Gaussian variable with zero mean
and the second-order moments
\[
	E_0[y^2(t)] = 0,
	\quad
	E_0[|y(t)|^2] = 1-e^{-\gamma t}.
\]
Thus the reference measure $\mu_0(dy(t))$
is given by
\begin{equation}
	\mu_0(dy(t))
	=  \frac{  e^{ - \frac{ |y(t)|^2 }{ 1-e^{-\gamma t} }  }   }{ \pi (1-e^{-\gamma t})  } d^2 y(t),
	\label{refm_het}
\end{equation}
where $d^2y= d( \mathrm{Re}y) d (\mathrm{Im} y) $.
From Eqs.~(\ref{povm1_het}) and (\ref{refm_het}),
the POVM for $y(t)$ is given by
\begin{equation}
	d^2 y(t)
	\mathscr{A} \left\{
	p(y(t)|\ha,\ha^\dagger;t)
	\right\} ,
	\notag
\end{equation}
where
\begin{equation}
	p(y(t)|\alpha,\alpha^\ast ; t)
	=
	\frac{
	\exp \left[
	-\frac{
	\left|
	y(t) - (1-e^{-\gamma t}) \alpha^\ast
	\right|^2
	}{
	e^{-\gamma t} (1-e^{-\gamma t})
	}
	\right]
	}{
	\pi e^{-\gamma t} (1-e^{-\gamma t})
	} .
	\label{conp_het}
\end{equation}
The probability distribution of the outcome $y(t)$
when the system is prepared in $\hrho_0$ at $t=0$
is given by
\begin{equation}
	p^Y_{\hrho_0}(y;t)
	= \int d^2 \alpha
	p(y(t)|\alpha,\alpha^\ast ; t)
	Q_{\hrho_0} (\alpha , \alpha^\ast),
	\label{yprob_het}
\end{equation}
where $Q_{\hrho} (\alpha, \alpha^\ast) := \bra{\alpha}\hrho\ket{\alpha}/\pi$
is the Q-function~\cite{1940264,PhysRev.138.B274},
and $\ket{\alpha}$ is a coherent state~\cite{PhysRev.131.2766}
defined by
\begin{equation}
	\ket{\alpha} 
	= e^{\alpha \ha^\dagger - \alpha^\ast \ha} \ket{0}
	=
	e^{-\frac{|\alpha|^2}{2} }
	\sum_{n=0}^\infty \frac{\alpha^n}{\sqrt{n!}} \ket{n}.
	\notag 
\end{equation}
From \equref{conp_het},
in the infinite-time limit $t \rightarrow \infty$,
the probability distribution of outcomes in \equref{yprob_het}
reduces to $Q_{\hrho_0}(y^\ast,y)$.
Thus the heterodyne measurement actually measures 
the non-commuting quadrature amplitudes simultaneously
in the sense that the probability distribution of outcomes is 
the Q-function of the initial state~\cite{carmichael2008statistical}.

As a reference POVM, we take
\begin{equation}
	d^2 \alpha \hE_{\alpha} = \frac{d^2\alpha}{ \pi }
	\ket{\alpha} \bra{\alpha}
	\label{xpovm_het}
\end{equation}
which generates the Q-function of the density operator.
From Eqs.~(\ref{mop_sol_het}) and (\ref{xpovm_het}) we have
\begin{equation}
	\mu_0(dy)
	\hM_y^\dagger(t)
	\hE_{\alpha} 
	\hM_y(t)
	=
	d^2y(t)
	q(\alpha , \alpha^\ast;y)
	\hE_{\tilde{\alpha}(\alpha, y)},
	\label{cond1_het}
\end{equation}
where
\begin{gather}
	\tilde{\alpha}(\alpha, y) 
	= e^{-\frac{\gamma t}{2}} \alpha + y^\ast,
	\label{talpha_het}
	\\
	q(\alpha, \alpha^\ast;y)
	=
	e^{-\gamma t} 
	p(y|\tilde{\alpha}(\alpha;y) , \tilde{\alpha}^\ast (\alpha;y) ).
	\label{qay_het}
\end{gather}
Note that the inferred quadrature amplitude in \equref{talpha_het}
allows a similar interpretation given in the homodyne analysis. 

Equation~(\ref{cond1_het}) ensures the condition in \equref{gen_cond}.
From Eqs.~(\ref{conp_het}), (\ref{talpha_het}) and (\ref{qay_het}),
for an arbitrary smooth function $F(\alpha, \alpha^\ast)$,
we have
\begin{align}
	&\int d^2 \alpha q(\alpha,\alpha^\ast;y) 
	F( \tilde{\alpha}(\alpha ; y),
	\tilde{\alpha}^\ast(\alpha;y))
	\notag \\
	&=\int d^2 \tilde{\alpha} (e^{\frac{\gamma t}{2}})^2
	q(e^{\frac{\gamma t}{2}} (\tilde{\alpha}+y^\ast),
	e^{\frac{\gamma t}{2}} (\tilde{\alpha}^\ast +y);y)
	F(\tilde{\alpha}, \tilde{\alpha}^\ast)
	\notag \\
	&=\int d^2\alpha
	p(y|\alpha,\alpha^\ast ;t) F(\alpha, \alpha^\ast).
	\notag
\end{align}
Thus, the condition~(\ref{gen_cond2}) 
for Theorem~\ref{th_rent_conservation1} is satisfied
and
the relative-entropy conservation law
\begin{equation}
	 D(P_{\hrho_0}^Y (\cdot;t) || P_{\hsigma_0}^Y (\cdot;t) )
	 = D_{\mathrm{Q}} (\hrho_0 || \hsigma_0) 
	 -E_{\hrho_0}[D_{\mathrm{Q}} ( \hrho_{y(t)}  ||  \hsigma_{y(t)} )]
	 \notag 
\end{equation}
holds,
where $\hrho_{y(t)}$ and $\hsigma_{y(t)} $ are the conditional
density operators for a given measurement outcome $y(t)$
and $D_{\mathrm{Q}}(\hrho||\hsigma)$
is the Q-function relative entropy defined as
\begin{equation}
	D_{\mathrm{Q}} (\hrho||\hsigma) 
	= \int d^2 \alpha Q_{\hrho}(\alpha , \alpha^\ast)
	\ln \left(
	\frac{ Q_{\hrho}(\alpha , \alpha^\ast) }{ Q_{\hsigma}(\alpha , \alpha^\ast) }
	\right).
	\label{qrel_def}
\end{equation}
Since the Q-function has the complete 
quantum information about the quantum state,
the Q-function relative entropy in \equref{qrel_def} 
vanishes if and only if
$\hrho = \hsigma$,
which is not the case in the diagonal relative entropies in
the preceding examples.
Still the Q-function relative entropy is bounded from above by 
the quantum relative entropy 
$S(\hrho||\hsigma) := \tr[ \hrho (\ln \hrho -\ln \hsigma  )]$,
for the relative entropy of probability distributions on the 
measurement outcome of a POVM
is always smaller than the quantum relative entropy~\cite{hayashi2006quantum}.

Equation~(\ref{qay_het}) implies the violation of Ban's condition~(\ref{ban_gen_cond}).
The difference of the Shannon entropies is given by
\begin{align}
	&H_{\hrho}(Q)
	-
	E_{\hrho} [  H_{\hrho_y}(Q) ]
	\notag
	\\
	&=
	H_{\hrho}(Q)
	+
	\int d^2 \alpha d^2 y
	e^{-\gamma t}
	p(y|\tilde{\alpha} (\alpha;y)) Q_{\hrho} (\tilde{\alpha}(\alpha;y))
	\ln \left(
		\frac{
			e^{-\gamma t }
			p(y|\tilde{\alpha} (\alpha;y)) 
			Q_{\hrho} (\tilde{\alpha}(\alpha;y))
		}{
			p^Y_{\hrho} (y)
		}
	\right)
	\notag
	\\
	&=
	- \gamma t
	+ I_{\hrho} (Q:Y)
	\neq
	I_{\hrho} (Q:Y)
	\label{het_last}
\end{align}
and the Shnnon-entropy conservation does not hold.
Again the the term $- \gamma t$ in \equref{het_last} 
originates from the non-unite Jacobian of the transformation
$x \rightarrow \tilde{x} (x;y)$.

\section{Summary}
\label{sec:conclusion}
In this paper we have examined the information flow
in a general quantum measurement process $Y$
concerning the relative entropy of the two quantum states
with respect to a system's POVM $X$ of the system.
By assuming the classicality condition on $X$ and $Y$,
we have proved 
the relative-entropy conservation law
when $X$ is a general POVM (Theorem~\ref{th_rent_conservation1})
and when $X$ is a PVM (Theorem~\ref{th_diag_conservation}).
The classicality condition
can be interpreted as the existence of a sufficient statistic
in a joint successive measurement of $Y$ followed by $X$
such that the distribution of the statistic coincides with
that of $X$ for the pre-measurement state.
This condition may be interpreted as a classicality condition 
because there exists a classical statistical model
which generates all the relevant probability distributions of $X$ and $Y$.
We have also investigated the case 
in which the labels of the PVM $X$ and the measurement outcome of $Y$
are both discrete
and we have shown the equivalence between
the classicality condition
in Theorem~\ref{th_diag_conservation} 
and the relative-entropy conservation law
for arbitrary states (Theorem~\ref{th_diag_equi}).
We have applied the general theorems to some typical 
quantum measurements.
In the QND measurement, the relative-entropy conservation law
can be understood as a result of 
the classical Bayes' rule 
which is a mathematical expression 
of the modification of our knowledge
based on the outcome of the measurement.
In the destructive sharp measurement of two-level systems,
Ban's condition together with the Shannon-entropy conservation law
does not hold, 
while our relative-entropy conservation law does.
The next examples,
namely
photon-counting, quantum counter, balanced homodyne and heterodyne measurements,
are non-QND measurements on a single-mode photon field
and the measurement outcomes
convey information about
the photon number,
part of the photon number,
one and both quadrature amplitude(s),
respectively.
In spite of the destructive nature of the measurements,
the classicality condition is still satisfied and 
we have shown that 
the relative-entropy conservation laws hold 
for these measurements.
In the quantum counter model, we can take two kinds of POVMs 
of the system
satisfying the two relative-entropy conservation laws.
In the heterodyne measurement
$X$ is the POVM which
generates the Q-function
and is not an ordinary PVM,
reflecting the fact that 
the non-commuting observables
are measured simultaneously.
In the examples of quantum counter, homodyne and heterodyne measurements,
the Shannon-entropy conservation laws do not hold 
due to the non-unit Jacobian of the transformation
$x \rightarrow \tilde{x} (x;y)$.
These examples of non-conserving Shannon entropies suggest that
our approach to the information transfer of the system's observable
is applicable to a wider range of measurement models
than that based on the Shannon entropy.

\begin{acknowledgments}
This work was supported by
KAKENHI Grant No.~26287088 from the Japan Society for the Promotion of Science, 
and a Grant-in-Aid for Scientific Research 
on Innovation Areas ``Topological Quantum Phenomena'' (KAKENHI Grant No.~22103005),
and the Photon Frontier Network Program from MEXT of Japan.
Y. K. acknowledges support 
by Advanced Leading Graduate Course for Photon Science (ALPS)
at the University of Tokyo.
\end{acknowledgments}

\appendix
\section{Equivalent conditions for (\ref{ban_gen_cond}) when $X$ and $Y$ are discrete}
\label{sec:app_ban}
In this appendix we characterize the condition~(\ref{ban_gen_cond}) required by Ban
when the reference POVM $X$ is a discrete PVM and the measurement $Y$ is also discrete.
In this case, the condition~(\ref{ban_gen_cond}) is equivalent to the condition that 
if a pre-measurement state is an eigenstate of $X$, 
then the post-measurement state is another eigenstate of $X$
as shown in the following theorem.
\begin{theo}
\label{th_d_bancond}
\begin{enumerate}
Let $\mathcal{E}^Y_y$ be a CP instrument with discrete measurement outcome $y$
and $\hE^X_x = \ket{x} \bra{x}$
satisfying the assumption~(\ref{gen_cond}) of Theorem~\ref{th_diag_conservation}.
Then the following conditions are equivalent:
\item \label{enum:d_ban}
Ban's condition~(\ref{ban_gen_cond}) holds,
i.e.
$q(x;y) = p(y|\tilde{x} (x;y))$.
\item \label{enum:d_delta}
For all $x$ and $y$ such that $p(y|x) \neq 0$,
\begin{equation}
	\sum_{x^\prime} \delta_{x, \tilde{x} (x^\prime ; y) } = 1.
	\label{d_delta}
\end{equation}
\item \label{enum:d_uni}
For all $x$ and $y$ such that $p(y|x) \neq 0$,
there exists a unique $x^\prime$ such that
$x = \tilde{x}(x^\prime;y)$.
\item \label{enum:d_r}
The post-measurement state is an eigenstate of $X$ if the pre-measurement state is an eigenstate.
Namely, for all $x$ and $y$,
there exist functions $\bar{x} (x;y)$ and $r(x;y) \geq 0$
such that
\begin{equation}
	 \mathcal{E}^Y_y  ( \ket{x} \bra{x} )
	= r(x;y) \ket{\bar{x}(x;y)} \bra{\bar{x}(x;y)}.
	\label{d_rcond}
\end{equation}
\end{enumerate}
\end{theo}
Before proving this theorem, we make a comment on the arbitrariness of the definition of $\tilde{x}(x;y)$
when $q(x;y)=0$.
In this case, $\tilde{x}(x;y)$ may take any value and we define it as
$\emptyset$, which is out of the range of label space of $X.$
We also define $p(y|\emptyset) = 0$ for any $y$.
\begin{proof}
$\ref{enum:d_ban} \Rightarrow \ref{enum:d_delta}$:
We first note that $p(y|x)$ in this case is given by
\equref{diag_pyx}.
By substituting $q(x^\prime ;y)= p(y|\tilde{x} (x^\prime;y))$ into \equref{diag_pyx},
we obtain
\begin{align}
	p(y|x) 
	= \sum_{x^\prime} \delta_{x,\tilde{x} (x^\prime;y)} p( y | \tilde{x}(x^\prime;y))
	= \left( \sum_{x^\prime} \delta_{x,\tilde{x} (x^\prime;y)} \right)
	p(y|x).
	\notag
\end{align}
Therefore Eq.~(\ref{d_delta}) holds whenever $p(y|x)\neq 0.$

The condition $\ref{enum:d_uni}$
immediately follows from \ref{enum:d_delta} by noting the definition of the Kronecker's delta.

$ \ref{enum:d_uni} \Rightarrow \ref{enum:d_r} $:
From Eq.~(\ref{gen_eydef}),
\begin{align}
	p(y|x) 
	= \tr \left[ \ket{x} \bra{x} {\mathcal{E}^Y_y}^\dagger (\hI) \right]
	= \tr \left[ \mathcal{E}^Y_y ( \ket{x} \bra{x} )  \right] .
	\label{d_pyx_eyx}
\end{align}
If $p(y|x) = 0,$
from Eq.~(\ref{d_pyx_eyx}) and 
the positivity of $\mathcal{E}^Y_y ( \ket{x} \bra{x} )$,
$\mathcal{E}^Y_y ( \ket{x} \bra{x} ) = 0$ and the condition~\ref{enum:d_r}
hold.
Let us consider the case in which $p(y|x)\neq 0$.
Since ${\mathcal{E}^Y_y}^\dagger$ is a CP map,
it has the following Kraus representation~\cite{Kraus1971311}
\begin{align}
	{\mathcal{E}^Y_y}^\dagger (\hA)
	=\sum_z \hM_{yz}^\dagger \hA \hM_{yz}.
	\label{app_kraus}
\end{align}
From Eq.~(\ref{gen_cond}),
we have
\begin{align}
	\sum_z \hM_{yz}^\dagger \ket{x} \bra{x} \hM_{yz}
	= q(x;y) \ket{\tilde{x} (x;y)} \bra{\tilde{x} (x;y)}.
	\notag
\end{align}
Therefore we can put
\begin{align}
	\hM_{yz}^\dagger \ket{x} 
	= a(x;y,z) \ket{ \tilde{x} (x;y)},
	\label{mdxax}
\end{align}
where
\begin{align}
	\sum_z |a(x;y,z)|^2 = q(x;y).
	\label{app_a2q}
\end{align}
From Eqs.~(\ref{app_kraus}) and (\ref{mdxax}) 
we obtain
\begin{align}
	{ \mathcal{E}^Y_y}^\dagger  ( \ket{x^{\prime \prime} }  \bra{x^{\prime}}  ) 
	&= \sum_z \hM_{yz}^\dagger \ket{x^{\prime \prime} }  \bra{x^{\prime}} \hM_{yz}
	\notag
	\\
	&= \left( \sum_{z}  a (x^{\prime \prime} ;y,z)  a^\ast (x^{\prime} ;y,z)  \right)
	\ket{ \tilde{x} (x^{\prime \prime} ;y) } \bra{ \tilde{x} (x^{\prime} ;y) }.
	\label{d_tochusiki1}
\end{align}
The matrix element of
$\mathcal{E}^Y_y (\ket{x} \bra{x})$
is evaluated as
\begin{align}
	\bra{x^\prime} \mathcal{E}^Y_y (\ket{x} \bra{x}) \ket{x^{\prime\prime}}
	&=
	\tr \left[ \mathcal{E}^Y_y (\ket{x} \bra{x})  \ket{x^{\prime\prime}}  \bra{x^\prime}   \right]
	\notag
	\\
	&=
	\tr \left[ \ket{x} \bra{x}  { \mathcal{E}^Y_y }^\dagger ( \ket{x^{\prime\prime}}  \bra{x^\prime}  ) \right]
	\notag
	\\
	&=
	\left( \sum_{z}  a (x^{\prime \prime} ;y,z)  a^\ast (x^{\prime} ;y,z)  \right)
	\delta_{x, \tilde{x} (x^{\prime \prime} ;y)}
	\delta_{x, \tilde{x} (x^{\prime } ;y)},
	\label{d_tochusiki2}
\end{align}
where we used \equref{d_tochusiki1} in the last equality.
From the condition~\ref{enum:d_uni}, 
there exists a unique $x^\prime$ such that
$x= \tilde{x} (x^{\prime} ;y)$
and we write this $x^\prime$ as
$\bar{x} (x;y)$.
Then Eq.~(\ref{d_tochusiki2}) becomes
\begin{align}
	\left( \sum_{z}  |  a (x^{\prime} ;y,z) |^2  \right)
	\delta_{x^\prime, \bar{x}(x;y)}
	\delta_{x^{\prime \prime} , \bar{x}(x;y)}
	= q(x^\prime ; y)
	\delta_{x^\prime, \bar{x}(x;y)}
	\delta_{x^{\prime \prime} , \bar{x}(x;y)},
	\label{app_qdd}
\end{align}
where we used \equref{app_a2q}.
Equation~(\ref{app_qdd}) implies
\begin{equation}
	\mathcal{E}^Y_y (\ket{x} \bra{x})
	=q (\bar{x} (x;y);y) \ket{ \bar{x} (x;y)} \bra{\bar{x} (x;y)},
	\notag
\end{equation}
which is nothing but the condition~\ref{enum:d_r}.

$\ref{enum:d_r}   \Rightarrow \ref{enum:d_ban}:$
From 
\[
	\hE^Y_y =  { \mathcal{E}^Y_y}^\dagger (\hI)
	=\sum_x p(y|x)\ket{x} \bra{x}
\] 
and
Eq.~(\ref{d_rcond}), we have
\begin{align}
	p(y|x)
	&= \tr[ \ket{x} \bra{x}  {\mathcal{E}^Y_y}^\dagger (\hI)  ]
	\notag
	\\
	&= \tr  [  \mathcal{E}^Y_y (  \ket{x} \bra{x} )   ]
	\notag
	\\
	&= r(x;y).
	\label{d_peqr}
\end{align}
From Eqs.~(\ref{gen_cond}), (\ref{d_rcond}) and (\ref{d_peqr}), we obtain
\begin{align}
	q(x;y) 
	&=
	\tr \left[
	\ket{\tilde{x} (x;y)} \bra{\tilde{x} (x;y)} 
	{\mathcal{E}^Y_y}^\dagger ( \ket{x} \bra{x} )
	\right]
	\notag
	\\
	&=
	\tr \left[
	\mathcal{E}^Y_y (
	\ket{\tilde{x} (x;y)} \bra{\tilde{x} (x;y)} 
	)
	\ket{x} \bra{x} 
	\right]
	\notag
	\\
	&=
	p(y| \tilde{x} (x;y)) \delta_{x, \bar{x} (\tilde{x} (x;y);y)}.
	\label{daplast}
\end{align}
When $q(x;y) \neq 0$, \equref{daplast} implies $q(x;y) = p(y| \tilde{x} (x;y))$.
If $q(x;y) = 0$,
$\tilde{x}(x;y) = \emptyset$ and $p(y|\emptyset) = q(x;y)=0$
from the remark above the present proof.
Thus the condition~(\ref{ban_gen_cond}) holds.
\end{proof}

We briefly remark on the case when the PVM $\ket{x}\bra{x}$ is continuous
with the complete orthonormal condition~(\ref{cons2}).
Under the same assumptions of Theorem~\ref{th_d_bancond},
we can show that Ban's condition~(\ref{ban_gen_cond})
implies
\begin{equation}
	\int dx^\prime 
	\delta  ( x - \tilde{x} (x^\prime;y) )
	= 1
	\label{c_delta}
\end{equation}
for any $x$ and $y$ such that $p(y|x) \neq 0.$
The proof of Eq.~(\ref{c_delta}) is formally as the same as
that of $\ref{enum:d_ban} \Rightarrow \ref{enum:d_delta}$
in Theorem~\ref{th_d_bancond}.
However, the formal correspondence between continous and discrete $X$ fails when
we consider the other part of the proof of Theorem~\ref{th_d_bancond}.
For example, we cannot conclude from \equref{c_delta} 
the existence and uniqueness of $x^\prime$ such that $\tilde{x}(x^\prime;y)= x$.
For simplicity let us assume the uniqueness of $x^\prime$ holds.
Still the condition~(\ref{c_delta}) is very restrictive since it implies
\begin{equation}
	\left|
	\frac{ \partial \tilde{x} (x^\prime ;y) }{  \partial x^\prime}
	\right|
	=1,
	\notag
\end{equation}
i.e. the Jacobian of the transformation $x \rightarrow \tilde{x}(x;y)$ should be 1.
This reflects the strong dependence of the Shannon entropy on the reference measure,
which is not the case in the relative entropy.

\section{Derivations of Eqs.~(\ref{mxxm}) and (\ref{povm_homo})}
\label{sec:app1}
To evaluate the operator
$\hM_{y(t)}^\dagger (t) \ket{x}_1 {}_1
\bra{x} \hM_{y(t)}(t)$,
we utilize the technique of normal ordering. 
We first note that 
the normally ordered expression
$:O(\ha , \ha^\dagger):$
of an operator $\hO$,
in which the annihilation operators are
placed to the right of the creation operators,
is given by a coherent-state expectation as
\begin{equation*}
	O(\alpha, \alpha^\ast) 
	= \bra{\alpha}  \hO  \ket{\alpha}.
\end{equation*}
Since the coherent state $\ket{\alpha}$ 
in the $\ket{x}_1$ representation is given by
\begin{equation*}
	{}_1\braket{x}{\alpha} 
	= \pi^{-1/4}
	\exp\left[
	-\frac{ 1 }{ 2 } (x- \sqrt{2} \alpha)^2
	-\frac{ 1 }{ 2 } (\alpha^2 + |\alpha|^2)
	\right],
\end{equation*}
we have 
\[
	\braket{\alpha}{x}_1 {}_1 \braket{x}{\alpha}
	= \pi^{-1/2}
	\exp \left[
	- 
	\left(  x -  \frac{ \alpha + \alpha^\ast }{\sqrt{2}}  \right)^2
	\right] ,
\]
which implies the following normally ordered expression
\begin{equation}
	\ket{x}_1 {}_1 \bra{x}
	= \pi^{-1/2}
	:
	\exp \left[
	- 
	\left(  x -  \frac{ \ha + \ha^\dagger}{\sqrt{2}}  \right)^2
	\right]
	: .
	\label{x_normal}	
\end{equation}
By using \equref{x_normal}
and the formula
\[
	e^{-\lambda \hn} \ket{ \alpha }
	=
	e^{ - \frac{|\alpha|^2 }{2} ( 1 - e^{ - 2\lambda } )  }
	\ket{  e^{-\lambda} \alpha},
\]
which is valid for real $\lambda$,
the expectation of
the operator $\hM_y^\dagger(t) \ket{x}_1 {}_1\bra{x} \hM_y(t)$
over the coherent state $\ket{\alpha}$ is evaluated to be 
\begin{align}
	&\bra{\alpha} 
	\hM_{y(t)}^\dagger (t) 
	\ket{x}_1
	{}_1
	\bra{x}
	 \hM_{y(t)}(t)
	\ket{\alpha}
	\notag \\
	&=
	\pi^{-1/2}
	\exp \left[
	-  \left(
		e^{-\frac{\gamma t}{2} } x 
		+ \frac{y(t)}{\sqrt{2}}
		- \frac{ \alpha + \alpha^\ast}{\sqrt{2}}
	   \right)^2
	   + \left(
	   e^{-\frac{\gamma t}{2} } x 
		+ \frac{y(t)}{\sqrt{2}}
		\right)^2
	  -x^2
	\right].
	\label{temp2_homo}
\end{align}
Substituting \equref{x_normal} 
in \equref{temp2_homo},
we obtain \equref{mxxm}.
By integrating \equref{mxxm}
with respect to $x$
and noting a relation
\[
	f(\hX_1) = \int dx f(x) \ket{x}_1 {}_1 \bra{x},
\]
which is valid for an arbitrary function $f(x)$,
we obtain
\begin{equation}
	\hM_y(t)^\dagger \hM_y(t)
	=
	\exp \left[
	\frac{\gamma t}{2}
	+\hX_1^2
	-e^{\gamma t}
	\left(\hX_1 -\frac{y}{\sqrt{2}}  \right)^2
	\right].
	\label{mymy_app1}
\end{equation}
To evaluate the proper POVM for the outcome $y$,
we need to multiply 
$\hM_y^\dagger(t) \hM_y(t)$ by $\mu_0(dy(t))$,
where $\mu_0(dy(t))$ is the probability measure
of $y(t)$,
provided that $\xi(\cdot)$ obeys a Wiener distribution.
Here $y(t)$ in \equref{y_homo}
under a Wiener measure $\mu_0$
is a Gaussian stochastic variable with the first and second moments
\[
	E_0[ y(t)] = 0,
\]
\[
	E_0[ y^2(t)]
	=
	\gamma \int_0^t e^{-\gamma s} ds 
	= 1 - e^{-\gamma t} ,
\]
where $E_0[\cdot]$ denotes the expectation
with respect to the Wiener measure.
Thus $\mu_0(dy(t))$ is given by
\begin{equation}
	\frac{dy}{\sqrt{2\pi (1-e^{- \gamma t})}}
	\exp \left[
	- \frac{
	y^2
	}{
	2 (1-e^{- \gamma t}) 
	}
	\right].
	\label{y_wiener}
\end{equation}
Multiplying Eq.~(\ref{mymy_app1})
by \equref{y_wiener},
we obtain \equref{povm_homo}.

%

\end{document}